\documentclass[letterpaper]{article} 
\usepackage{aaai24}  
\usepackage{times}  
\usepackage{helvet}  
\usepackage{courier}  
\usepackage[hyphens]{url}  
\usepackage{graphicx} 
\usepackage{natbib}  
\usepackage{caption} 
\usepackage{algorithm}
\usepackage{algorithmic}

\nocopyright

\usepackage{newfloat}
\usepackage{listings}
\DeclareCaptionStyle{ruled}{labelfont=normalfont,labelsep=colon,strut=off} 
\floatstyle{ruled}
\newfloat{listing}{tb}{lst}{}
\floatname{listing}{Listing}
\title{Participation Incentives in Approval-Based Committee Elections}
\author {
    Martin Bullinger\textsuperscript{\rm 1}\thanks{Most of this research was done when Martin Bullinger was a PhD student at the Technical University of Munich.},
    Chris Dong\textsuperscript{\rm 2},
    Patrick Lederer\textsuperscript{\rm 2}, 
    Clara Mehler\textsuperscript{\rm 2}
}
\affiliations {
    \textsuperscript{\rm 1}Department of Computer Science, University of Oxford\\
    \textsuperscript{\rm 2}School of Computation, Information and Technology, Technical University of Munich\\
    martin.bullinger@cs.ox.ac.uk, chris.dong@tum.de, ledererp@in.tum.de, mehler@in.tum.de
}

    \usepackage{fix-cm}
    \usepackage{thm-restate} 
	\usepackage{xspace}
	\usepackage{amsthm}
    \usepackage{amsmath}
    \usepackage{amsfonts}
    \usepackage{amssymb}

	\usepackage{nicefrac}
	\usepackage{tikz}
	\usetikzlibrary{positioning,arrows,shapes,calc}
	\usepackage{colortbl}
	\usepackage{booktabs}
	\usepackage{cleveref}
	\usepackage{array}
	\usepackage{graphicx}
	\usepackage{xcolor}
	\usepackage{enumitem}
	\usepackage{tabularx}
	\usepackage{thm-restate}
	\usepackage{ragged2e}
    \usepackage{microtype}
    \usepackage{cleveref}
    \usepackage{complexity}
    \usepackage[backgroundcolor=orange!15!white,bordercolor=orange!80!black,textsize=small]{todonotes} 
\presetkeys{todonotes}{inline}{}

    \def\multiset#1#2{\ensuremath{\left(\kern-.3em\left(\genfrac{}{}{0pt}{}{#1}{#2}\right)\kern-.3em\right)}}

	\theoremstyle{definition}

	\theoremstyle{plain}
	\newtheorem{theorem}{Theorem}

	\newtheorem{corollary}{Corollary}

    \definecolor{TUMblue}{RGB}{0,101,189}
    \definecolor{TUMgreen}{HTML}{A2AD00} 

\newcommand{\sphrag}{\texttt{seqPhragm{\'e}n}\xspace}
\newcommand{\mes}{\texttt{MES}\xspace}

\newcommand{\pav}{\texttt{PAV}\xspace}
\newcommand{\spav}{\texttt{seqPAV}\xspace}
\newcommand{\ccav}{\texttt{CCAV}\xspace}
\newcommand{\sccav}{\texttt{seqCCAV}\xspace}
\newcommand{\av}{\texttt{AV}\xspace}

\newcommand{\sco}[3][A]{\mathrm{score}_{#2}(#1,#3)} 

\newcommand{\convexpath}[2]{
[   
    create hullnodes/.code={
        \global\edef\namelist{#1}
        \foreach [count=\counter] \nodename in \namelist {
            \global\edef\numberofnodes{\counter}
            \node at (\nodename) [draw=none,name=hullnode\counter] {};
        }
        \node at (hullnode\numberofnodes) [name=hullnode0,draw=none] {};
        \pgfmathtruncatemacro\lastnumber{\numberofnodes+1}
        \node at (hullnode1) [name=hullnode\lastnumber,draw=none] {};
    },
    create hullnodes
]
($(hullnode1)!#2!-90:(hullnode0)$)
\foreach [
    evaluate=\currentnode as \previousnode using \currentnode-1,
    evaluate=\currentnode as \nextnode using \currentnode+1
    ] \currentnode in {1,...,\numberofnodes} {
-- ($(hullnode\currentnode)!#2!-90:(hullnode\previousnode)$)
  let \p1 = ($(hullnode\currentnode)!#2!-90:(hullnode\previousnode) - (hullnode\currentnode)$),
    \n1 = {atan2(\y1,\x1)},
    \p2 = ($(hullnode\currentnode)!#2!90:(hullnode\nextnode) - (hullnode\currentnode)$),
    \n2 = {atan2(\y2,\x2)},
    \n{delta} = {-Mod(\n1-\n2,360)}
  in 
    {arc [start angle=\n1, delta angle=\n{delta}, radius=#2]}
}
-- cycle
}

\allowdisplaybreaks
\begin{document}

\maketitle

\begin{abstract}
    In approval-based committee (ABC) voting, the goal is to choose a subset of predefined size of the candidates based on the voters' approval preferences over the candidates. 
    While this problem has attracted significant attention in recent years, the incentives for voters to participate in an election for a given ABC voting rule have been neglected so far. 
    This paper is thus the first to explicitly study this property, typically called participation, for ABC voting rules.
    In particular, we show that all ABC scoring rules even satisfy group participation, whereas most sequential rules severely fail participation.
    We furthermore explore several escape routes to the impossibility for sequential ABC voting rules: we prove for many sequential rules that \emph{(i)} they satisfy participation on laminar profiles, \emph{(ii)} voters who approve none of the elected candidates cannot benefit by abstaining, and \emph{(iii)} it is $\NP$-hard for a voter to decide whether she benefits from abstaining.
\end{abstract}

\section{Introduction} 

Many questions in multi-agent systems reduce to the problem of selecting a subset of the available candidates based on the preferences of a group of agents over these candidates. Maybe the most apparent example for this are elections of parliaments or city councils, but there are also numerous applications beyond classical voting. 
For instance, this model can also be used to describe automated recommender systems \citep[][]{GaFa22a} or the selection of validators in a block chain \citep[][]{CeSt21a}. 
In the field of computational social choice, such elections are known as approval-based committee (ABC) elections and they have recently attracted significant attention \citep[][]{FSST17a,LaSk22b}. 
In more detail, this research studies \emph{approval-based committee (ABC) voting rules}, which choose a fixed-size subset of the candidates, typically called a \emph{committee}, based on the voters' approval ballots (i.e., voters express their preference about every candidate by either approving or disapproving her).

One of the basic premises of ABC voting rules (and, more generally, of all types of elections) is that voters will participate in the election. However, this is not necessarily in the interest of the voters: for example, for many single-winner voting rules, there are situations where voters prefer the outcome chosen when abstaining to the outcome chosen when voting \citep[e.g.,][]{Moul88b,Pere01a,BBGH18a}. 
This undesirable phenomenon, which is known as the \emph{no-show paradox}, entails that voting can be disadvantageous for a voter and hence disincentivizes participation. We are thus interested in voting rules that avoid this paradox, which are then said to satisfy \emph{participation}. 
Note that while related concepts have been analyzed \citep[e.g.,][Prop. A.3]{SaFi19a,LaSk22b}, participation has not been studied for ABC voting rules and we thus initiate the study of this axiom for ABC elections.

\paragraph{Our contribution.}
In this paper, we study the participation incentives of ABC voting rules. In more detail, we first investigate which ABC voting rules satisfy participation and prove that all ABC scoring rules (including all Thiele rules) even satisfy group participation. 
This generalizes the observation that scoring rules satisfy participation for single-winner elections and gives a strong argument in favor of Thiele rules. 
By contrast, we prove a general impossibility theorem, which shows that most ABC voting rules that sequentially compute the winning committees fail participation. 
In particular, our result implies that all sequential Thiele rules (except for approval voting) as well as the sequential variant of Phragm\'en's rule and the method of equal shares severely fail participation: there are situations where a voter only approves one of the elected candidates when she votes but all except one of the elected candidates when she abstains. 
These theorems also subsume results by \citet{SaFi19a} and \citet{LaSk22b} who study a monotonicity axiom that constitutes a special case of participation.

We furthermore analyze several approaches to circumvent this negative result for sequential rules. Firstly, motivated by the notion of strategyproofness for unrepresented voters by \citet{DDE+22a}, we show that many sequential ABC voting rules ensure that voters who do not approve any of the elected candidates cannot benefit by abstaining (\Cref{prop:unrep}). This result complements our impossibility theorem, which shows that voters can significantly benefit from abstaining when they approve at least one candidate and, moreover, demonstrates that many sequential rules satisfy at least a minimal degree of participation. 
Next, we prove that all sequential Thiele rules and the sequential Phragm\'en rule satisfy participation when restricting the domain to laminar profiles (\Cref{prop:laminar}). 
These profiles have been introduced by \citet{PeSk20a} and require that for all candidates $x$ and $y$, the respective sets of voters approving $x$ and~$y$ are either disjoint or related by subset inclusion. 
Hence, this result shows that sequential ABC voting rules satisfy participation when focusing on an important special case.

Finally, we show that it is \NP-hard for a voter to decide whether she benefits from abstaining when using sequential Thiele rules, sequential Phragm\'en, or the method of equal shares (\Cref{subsec:hardness}). Thus, even though a voter may benefit by abstaining, she may not be able to recognize it. Moreover, our technique for showing these hardness results is very universal and allows us to recover, strengthen, or complement existing hardness results by \citet{FGK22b} and \citet{JaFa23b}. 
In addition, our results indicate that many basic problems (e.g., whether there is a winning committee for which a given voter approves $\ell$ candidates) are \NP-hard for sequential rules.

\paragraph{Related Work.} The topic of ABC voting currently attracts significant attention and we refer to \citet{LaSk22b} for a recent survey. 
While there is, to the best of our knowledge, no explicit work on participation in ABC voting, there are a few closely related papers.
In particular, \citet{SaFi19a} study an axiom called support monotonicity with population increase (SMWPI), which requires that the abstention of a voter cannot result in a committee that contains all of her approved candidates if such a committee is not chosen when voting. 
Clearly, SMWPI is a mild variant of participation and \citet{SaFi19a} show that all ABC scoring rules satisfy this condition. 
Moreover, \citet[Prop. A.3]{LaSk22b}, \citet{MoOl15a}, and \citet{Jans16a} consider various sequential rules and prove that they fail SMWPI, which implies that they also fail participation. 
Notably, the proof of \Cref{thm:iterativeNoShow} also works with SMWPI and our result thus strengthens the existing results by showing that essentially all sequential rules fail this property. 

Our paper is also related to the study of strategyproofness and robustness in ABC voting \citep[e.g.,][]{AGG+15a,Pete18a,BFK+21a,FGK22b}. 
In particular, participation can be seen as a variant of strategyproofness that prohibits that voters manipulate by abstaining, or as a robustness axiom that measures how much impact an abstaining voter can have on the outcome. 
Many of these papers are conceptually similar to ours as they first study whether ABC voting rules satisfy an axiom and then explore escape routes.

Finally, in the broader realm of social choice, there are numerous papers that study participation. 
In his seminal paper, \citet{Moul88b} showed that a large class of single-winner voting rules known as Condorcet extensions fail participation. 
This result caused a large amount of follow-up work, which either strengthens the negative result \citep[e.g.,][]{JPG09a,Dudd14b,BGP16c} or explores escape routes \citep[e.g.,][]{BBH15c,BBGH18a}. 
A particularly noteworthy paper in our context is by \citet{PJG10a} who show that a large class of committee voting rules fail participation when voters report ranked ballots. 
We refer to \citet{Hofb19a} for a survey on participation in social choice.

\section{Preliminaries}

In this paper, we will use the standard ABC voting setting following the notation of \citet{LaSk22b}. To formalize this model, let $\mathbb{N}$ denote an infinite set of voters and let $C$ denote a set of $m>1$ candidates. An electorate $N$ is a non-empty and finite subset of $\mathbb{N}$ and we suppose that every voter $i\in N$ reports an approval ballot $A_i$ to express her preferences. Formally, an approval ballot is a non-empty subset of $C$. An \emph{approval profile $A$} is the collection of the approval ballots of all voters $i\in N$, i.e., a function of the type $N\rightarrow 2^C\setminus \{\emptyset\}$. We denote by $N_A$ the set of voters that report a ballot in profile $A$ and by $N_A(c)$ the set of voters who approve candidate $c$ in $A$. Moreover, $A_{-i}$ (resp. $A_{-I}$) is the profile derived from $A$ when voter $i\in N_A$ (resp. a group of voters $I\subseteq N_A$) abstains. More formally, $A'=A_{-i}$ is defined by $N_{A'}=N_A\setminus \{i\}$ and $A_j'=A_j$ for all $j\in N_{A'}$.

Given an approval profile, our goal is to elect a \emph{committee}, which is a subset of the candidates of predefined size. 
Following the literature, we define $k\in \{1,\dots,m-1\}$ as the target committee size and $\mathcal{W}_k=\{W\subseteq C\colon |W|=k\}$ as the set of size $k$ committees. 
We collect all information associated with an election in an \emph{election instance $E=(N,C,A,k)$}, where we omit $N$ and $C$ whenever they are clear from the context. 
Given an election instance $E$, our goal is to determine the winning committee. 
To this end, we will use \emph{approval-based committee (ABC) voting rules} which map every election instance $E$ to a non-empty subset of $\mathcal{W}_k$, i.e., ABC voting rules may return multiple committees that are tied for the win.

\subsection{Classes of Voting Rules}

We now introduce several (classes of) ABC voting rules. We assume that all rules return all committees that can be obtained by some tie-breaking order.

\paragraph{ABC scoring rules.}
ABC scoring rules, which were introduced by \citet{LaSk21a}, generalize scoring rules to ABC elections: each voter gives points to each committee and the winning committees are those with the maximal total score. Formally, these rules are defined by a scoring function $s$ which maps all $x,y\in \mathbb N_0$ with $x\leq y$ to a 
rational number $s(x,y)$ such that $s(x,y)\geq s(x',y)$ for all $x'\leq x\leq y$. 
Without loss of generality, we suppose that $s(0,y)=0$ for all $y$. Intuitively, $s(x,y)$ is the score a voter gives to a committee $W$ when she approves $x$ members of $W$ and $y$ in total. Thus, the total score of a committee $W$ in a profile $A$ is $\hat s (A,W) := \sum_{i\in N_A} s(| A_i \cap W |, |A_i|)$. 
The ABC scoring rule defined by the scoring function $s$ chooses the committees $W$ that maximize the total score $\hat s(A,W)$.

\paragraph{Thiele rules.} Thiele rules, suggested by \citet{Thie95a}, are scoring rules that are independent of the ballot size, i.e., $s(x,y)=s(x,y')$ for all $x\le y\le y'$.
Therefore, we drop the second argument of the scoring function.
We impose the standard requirements that $s(1)>0$ and $s(x+1)-s(x)\geq s(x+2)-s(x+1)$ for all $x\in\mathbb{N}_0$ (\emph{concavity}).
Important examples of Thiele rules are multiwinner approval voting (\texttt{AV}), defined by $s(x) = {x}$, proportional approval voting (\texttt{PAV}), defined by $s(x) = \sum_{y= 1}^x \frac{1}{y}$, and Chamberlin-Courant approval voting (\texttt{CCAV}), defined by $s(x) = 1$ for all $x>0$. 

\paragraph{Sequential query rules.}
Generalizing concepts of \citet{BDI+23a} and \citet{DoLe22a}, we introduce the class of \emph{sequential query rules}. The idea of this class is to encapsulate ABC voting rules that compute the winning committees by sequentially adding candidates. To formalize this, we let $\mathcal{S}(C)$ denote the set of all non-repeating sequences of candidates with length $\ell\leq m-2$. In particular, the empty set is the only sequence of length $0$. 
The central concept for sequential query rules are \emph{query functions}~$g$ which take a profile $A$, a target committee size $k$, and a sequence $S=(c_1,\dots,c_\ell)\in\mathcal S(C)$ as input and return a subset of $C\setminus \{c_1,\dots,c_\ell\}$. 
Intuitively, $g(A,k,S)$ are the candidates that will be chosen next given that the candidates in $S$ have been selected in this order. Moreover, we demand that $g(A,k,S)$ is non-empty whenever $S$ is generated by $g$. Formally, we say a sequence $S=(c_1,\dots, c_\ell)$ is \emph{valid} for $g(A,k,\cdot)$ if $S=\emptyset$ or $c_i\in g(A, k, (c_1,\dots, c_{i-1}) )$ for all $1\leq i\leq \ell$. We require that $g(A,k,S)\neq\emptyset$ whenever $S$ is valid and $\ell<k$.
Finally, an ABC voting rule $f$ is a \emph{sequential query rule} induced by the query function $g$ if $f(A,k)=\{ \{c_1,\dots, c_k\}\in \mathcal W_k \colon \text{ $(c_1,\dots, c_k)$ is valid for $g(A,k,\cdot)$} \}$ for all profiles $A$ and committee sizes $k$. The class of sequential query rules as defined here is actually equivalent to the set of ABC voting rules as there are no restrictions on $g$. Hence, we will later introduce axioms for sequential query rules to pinpoint when a sequential query rule fails participation. In the following, we introduce several voting rules that can be easily described as sequential query rules.

\paragraph{Sequential Thiele rules.}
Sequential Thiele rules are greedy versions of Thiele rules and have also been suggested by \citet{Thie95a}. 
Given some Thiele scoring function $s$, these rules extend in every step each committee $W$ of the previous step with the candidates $c$ that increase the score the most. 
More formally, sequential Thiele rules are sequential query rules defined by the query function $g(A,k,(c_1,\dots,c_\ell))=\text{argmax}_{x\in C\setminus \{c_1,\dots, c_\ell\}} \hat s(A, \{x,c_1,\dots,c_\ell\})$.
Prominent examples of sequential Thiele rules are {\sccav} and {\spav} which are the sequential versions of {\ccav} and \pav. Note that the sequential version of {\av} is identical to \av.

\paragraph{Sequential Phragmén.} This rule (\sphrag), which was suggested by \citet{Phra95a} and rediscovered by \citet{BFJL16a}, relies on a cost-sharing mechanism. In more detail, \sphrag assumes that each candidate has a cost of $1$ and each voter starts with a budget of $0$. Over time, the budget of each voter increases uniformly and as soon as the voters that approve some candidate $c$ have a total budget of $1$, they buy $c$ and add it to the winning committee. The budget of these (and only these) voters is then reset to $0$. The process continues until $k$ candidates have been bought. Clearly, {\sphrag} is a sequential query rule. 

\paragraph{Method of equal shares.}
The method of equal shares (\mes), which is due to \citet{PeSk20a}, works similar to \sphrag. In particular, every candidate again costs $1$, but every voter~$i$ starts with a budget of $x_0(i)=\frac{k}{n}$ instead of earning budget over time. \mes then tries to buy candidates in sequential steps.
In more detail, let $x_r(i)$ denote the budget of each voter $i$ after $r$ steps and let $X=\{c_1,\dots, c_r\}$ denote the set of candidates that have already been bought. 
We define by $C_r:=\{c\in C\setminus X\colon \sum_{i\in N_A(c)} x_r(i)\geq 1\}$ the set of candidates that can still be afforded. If $C_r\neq\emptyset$, we add the candidate $c\in C_r$ to the winning committee that incurs the minimal cost to the voter paying the most when splitting the cost as equally as possible, i.e., $c$ minimizes $\rho(c)$ with $\sum_{i\in N_A(c)} \min(\rho(c), x_r(i))=1$. 
Next, we set $x_{r+1}(i)=x_r(i) - \min(\rho(c), x_r(i))$ for $i\in N_A(c)$ 
and $x_{r+1}(i)=x_{r}(i)$, otherwise. 
We then continue with the next round. 
This process, typically called Phase~1 of \mes, iterates until $C_r=\emptyset$. 
If at this point less than $k$ candidates have been bought, Phase~2 of \mes starts where we have to complete the committee. For this, a variant of \sphrag is used where voters keep their remaining budget from Phase~1. 

\subsection{Participation}

We next turn to the central axiom of this paper, participation. The idea of this condition is that voters should not be worse off when voting instead of abstaining. 
To formalize this, we say that a voter $i$ (weakly) \emph{prefers} a committee $W$ to committee $W'$ (denoted by $W\succsim_i W'$) if $\lvert W\cap A_i\rvert \geq \lvert W'\cap A_i\rvert$, and \emph{strictly prefers} $W$ to $W'$ (denoted by $W\succ_i W'$) if $\lvert W\cap A_i\rvert > \lvert W'\cap A_i\rvert$. 
This approach is the standard to extend voters' preferences to preferences over committees \citep[see, e.g.,][]{AGG+15a,Bota21a,DDE+22a}. 
Since our ABC voting rules return sets of committees, we furthermore need to lift the voters' preferences to sets of committees. 
Following the literature \citep{KdV+20a,Bota21a}, we use \emph{Kelly's extension}. 
This extension states that a voter $i$ prefers a set of committees $X$ to another set of committees $Y$ (denoted by $X\succsim_i Y$) if $W\succsim_i W'$ for all $W\in X$ and $W'\in Y$ \citep{Kell77a}. 
Moreover, this preference is strict (denoted by $X\succ_i Y$) if there are $W\in X$, $W'\in Y$ with $W\succ_i W'$. Kelly's extension guarantees that $X\succsim_i Y$ if and only if voter $i$ weakly prefers the outcome chosen from $X$ to the outcome chosen from $Y$ regardless of the tie-breaking. 
We note, however, that all of our results except the complexity results in \Cref{subsec:hardness} are rather independent of the extension to sets of committees and, for instance, also hold under lexicographic tie-breaking or other extensions such as Fishburn's extension \citep{Fish72a}, Gärdenfors' extension \citep{Gard76a} or the leximax extension \citep{JPG09a}.

Now, an ABC voting rule $f$ satisfies \emph{participation} if $f(A_{-i},k)\not\succ_i f(A,k)$ for all profiles $A$, voters $i\in N_A$, and committee sizes $k\in \{1,\dots,m-1\}$. 
Put differently, participation ensures that voters can never benefit by abstaining. 
To further strengthen the axiom, we say that a group of voters $I\subsetneq N_A$ \emph{benefits from abstaining} for a profile $A$ and committee size $k$ if $f(A_{-I},k)\succsim_i f(A,k)$ for all $i\in I$ and $f(A_{-I},k)\succ_i f(A,k)$ for some $i\in I$. Then, an ABC voting rule $f$ satisfies \emph{group participation} if it is never possible for a group of voters to benefit by abstaining.

\section{Results}

We are now ready to formulate our results. 
In \Cref{subsec:part}, we will show that ABC scoring rules satisfy group participation and that most sequential ABC voting rules fail participation. In \Cref{subsec:axiomaticEscapeRoutes}, we thus explore two axiomatic escape routes to the impossibility for sequential rules. 
Finally, in \Cref{subsec:hardness}, we show for our considered sequential ABC voting rules that it is \NP-hard to decide for a voter whether she benefits by abstaining. Due to space restrictions, we defer all proofs not discussed in this section to the appendix.  

\subsection{Participation for ABC Voting Rules}\label{subsec:part}

The goal of this section is to understand which ABC voting rules satisfy participation. To this end, we first show that all ABC scoring rules even satisfy group participation. 

\begin{theorem}\label{thm:scoringPossibility}
    Every ABC scoring rule satisfies group participation.
\end{theorem}
\begin{proof}
    Let $f$ be an ABC scoring rule and let $s$ denote its scoring function. We assume for contradiction that there is a profile $A$, a committee size $k$, and a group of voters $I\subsetneq N_A$ that benefits from abstaining, i.e., $f(A_{-I},k)\succsim_i f(A,k)$ for all $i\in I$ and $f(A_{-I},k)\succ_{i^*} f(A,k)$ for some $i^*\in I$. Next, we proceed with a case distinction with respect to $f(A_{-I},k)$ and first consider the case that there is $W\in f(A_{-I},k)\setminus f(A,k)$. 
    By definition of Kelly's extension, this means that $|W\cap A_i|\geq |W'\cap A_i|$ for all $W'\in f(A,k)$ and $i\in I$. Using the definition of ABC scoring rules, it hence follows that $\sum_{i\in I} s(|A_i\cap W|, |A_i|)\geq \sum_{i\in I} s(|A_i\cap W'|, |A_i|)$. On the other hand, it holds that $\hat s(A,W')>\hat s(A,W)$ since $W'\in f(A,k)$ and $W\not\in f(A,k)$. Combining these facts then implies that $\hat s(A_{-I},W)=\hat s(A,W)-\sum_{i\in I} s(|A_i\cap W|, |A_i|)<\hat s(A,W')-\sum_{i\in I} s(|A_i\cap W'|, |A_i|)=\hat s(A_{-I}, W')$, which contradicts that $W\in f(A_{-I},k)$. 
    As the second case, we suppose that $f(A_{-I},k)\subsetneq f(A,k)$ and let $W\in f(A_{-I},k)$, $W'\in f(A,k)\setminus f(A_{-I},k)$. Using Kelly's extension, we infer again that $|W\cap A_i|\geq |W'\cap A_i|$ for all $i\in I$, so $\sum_{i\in I} s(|A_i\cap W|, |A_i|)\geq \sum_{i\in I} s(|A_i\cap W'|, |A_i|)$. 
    Moreover, $\hat s(A,W)=\hat s(A,W')$ because $W,W'\in f(A,k)$. 
    We conclude that $\hat s(A_{-I}, W)\leq \hat s(A_{-I},W')$, which contradicts our assumption because $W\in f(A_{-I},k)$ is then only possible if $W'\in f(A_{-I}, k)$. 
    Since we have a contradiction in both cases, $f$ satisfies group participation.
 \end{proof}

Next, we turn to sequential ABC voting rules. As discussed before, for some of these rules (e.g., \spav, \sphrag, and \mes), earlier results in the literature imply that they fail participation \citep[Prop. A.3]{Jans16a,LaSk22b}. We will next show that the incompatibility is much more far-reaching as essentially all sequential rules other than \av fail participation. In more detail, we introduce three mild axioms for query functions and prove that every sequential query rule whose query function satisfies these conditions fails participation. 

\paragraph{Continuity.} Continuity has been introduced by \citet{Youn75a} for single-winner voting rules and requires that large groups of voters can enforce that some of their desired outcomes are chosen. 
Formally, we say that a query function $g$ satisfies \emph{continuity} if for all $A, A',k$ and $S\in \mathcal S(C)$, it holds that $g(\lambda A + A',k,S)\subseteq g(A,k,S)$ for every sufficiently large $\lambda\in \mathbb N$. 
The sum of approval profiles means that a copy for each voter in each profile is present, and the multiplication by a non-negative integer $\lambda$ that $\lambda$ copies of each voter are present. 

\paragraph{Standardness.}
A condition that almost all commonly considered sequential ABC voting rules satisfy is that they typically choose the approval winner as first candidate. 
We formalize this as standardness: a query function $g$ is \emph{standard} if $g(A,k,\emptyset) = \av(A,1)$. 

\paragraph{Concurrence.} While the previous two axioms are necessary for \Cref{thm:iterativeNoShow}, they do not capture its essence 
as \av and other sequential rules satisfy continuity, standardness, and participation. 
The crucial observation is that sequential query rules only optimize the voter satisfaction myopically, which causes a dependence on the history of previous choices. 
For instance, consider the following profile with $4$ ballots, let $k=2$, and suppose that candidate $c$ is chosen first. 
\begin{equation*}
  1\times \{a,b\} \quad 1\times \{b,c\} \quad 1\times \{a\} \quad 1\times \{c\} 
\end{equation*}
Essentially all commonly considered sequential ABC voting rules but \av will choose $a$ and not $b$
as next candidate because the voter who approves $b$ and $c$ is already partially satisfied. 
More generally, let $S\in \mathcal S(C)$ be a sequence of already chosen candidates and consider $a,b\in {C}$ with $|N_A(a)|\geq |N_A(b)|$ such that in each step of $S$ the voters approving $b$ are at least as satisfied as the voters approving $a$.
Concurrence captures the idea that sequential rules prefer to choose $a$ over $b$ in such a situation.
To formalize this in a fashion that encompasses all commonly considered sequential rules, we add technical restrictions which weaken the axiom. In more detail, we call a query function $g$ \emph{concurring} if the following holds: 
Consider a profile $A$ such that $|A_i|\leq 2$ for all $i\in N_A$ and $|N_A(c)|=|N_A(d)|\leq \frac{n}{k}$ for all $c,d\in C$, and a sequence of already chosen candidates $(c_1,\dots, c_\ell)\in \mathcal{S}(C)$. 
Then, for all candidates $c,d\in C\setminus \{c_1,\dots,c_\ell\}$ such that $|\{i\in N_A\colon A_i=\{c_j,d\}\}|\geq |\{i\in N_A\colon A_i=\{c_j,c\}\}|$ for all $j\in \{1,\dots,\ell\}$, where one of these inequalities is strict, it holds that $d\not\in g(A,k,(c_1,\dots,c_\ell))$.
As we show next, myopic efficiency (in the form of concurrence) is the main culprit for the no-show paradox of sequential query rules.
\medskip

\begin{theorem}\label{thm:iterativeNoShow}
Every sequential query rule fails participation if $k\ge 3$ and its query function is standard, concurring, and continuous. 
Even more, a voter can obtain only $1$ approved candidate when participating while obtaining $k-1$ approved candidates when abstaining.
\end{theorem}

\begin{proof}
Let $f$ denote a sequential query rule and suppose its query function $g$ is standard, concurring, and continuous. Moreover, we let $k\ge 3$ denote the target committee size and set $C=\{a_1,\dots, a_r, b_1,\dots, b_r\}$ for $r=k-1$. Next, consider the following profile $A$: first, we add for every two element subset $B\subseteq C$ except for $\{a_1,b_1\}$ and $\{a_r,b_r\}$ a voter who approves $B$; 
second, we add for each $i\in \{2,\dots, r\}$ two voters who approve $\{a_i,b_1\}$ and for each $j\in \{2,\dots,r-1\}$ two voters who approve $\{b_j,a_r\}$; third, we add a voter who approves $\{b_1,b_r\}$ and a voter who approves $\{a_1,a_r\}$; finally, we add voters who approve only a single candidate such that all candidates have the same approval score. Moreover, if $k=3$, we add another candidate $d$ before the last step that shares three ballots $\{b,x\}$ with each $x\in C$ to ensure that $|N_A(x)|\leq \frac{n}{k}$ for all $x\in C$. This candidate will be ignored from now on as it does not affect our analysis. \Cref{fig:AlmostCompleteGraph} visualizes this profile by depicting ballots of size $2$ as edges.
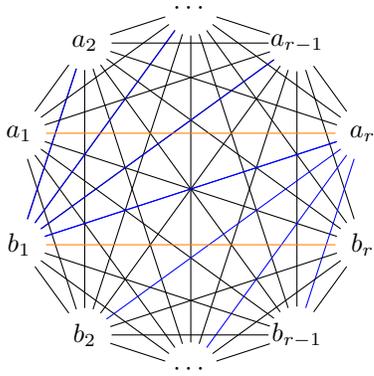
\begin{figure}
    \centering
    \begin{tikzpicture}[scale=0.80, every node/.style={scale=1}]

\def\radius{3} 
\def\startangle{162} 
\def\incrementangle{-36} 

\foreach \i/\text [count=\n from 0] in {1/$a_1$, 2/$a_2$, 3/$\cdots$, 4/$a_{r-1}$, 5/$a_r$} {
    \node[circle, inner sep = 0, minimum size = 2em] (a\i) at (\startangle+\n*\incrementangle:\radius) {\text};
}

\foreach \i/\text [count=\n from 0] in {1/$b_1$, 2/$b_2$, 3/$\cdots$, 4/$b_{r-1}$, 5/$b_r$} {
    \node[circle, inner sep = 0, minimum size = 2em] (b\i) at (\startangle+36 -\n*\incrementangle:\radius) {\text};
}

\foreach \i in {1,...,5} {
    \foreach \j in {1,...,5} {
        \ifnum\i=1
            \ifnum\j=1
            \else
                \draw (a\i) -- (b\j);
            \fi
        \else
            \ifnum\i=5
                \ifnum\j=5
                \else
                    \draw[blue] (a\i) -- (b\j); 
                \fi
            \else
                \draw (a\i) -- (b\j);
            \fi
        \fi
    }
}

\foreach \i in {2,...,5} {
    \draw[blue] (a\i) -- (b1); 
}

\foreach \i in {1,...,5} {
    \foreach \j in {\i,...,5} {
        \ifnum\i=\j
        \else
            \ifnum\i=1
                \ifnum\j=5 
                    \draw[orange] (a\i) -- (a\j); 
                \else \draw (a\i) -- (a\j); 
                \fi
                    
            \else    
                    \draw (a\i) -- (a\j); 
            \fi
        \fi
    }
}

\foreach \i in {1,...,5} {
    \foreach \j in {\i,...,5} {
        \ifnum\i=\j
        \else
            \ifnum\i=1
                \ifnum\j=5 
                    \draw[orange] (b\i) -- (b\j); 
                \else \draw (b\i) -- (b\j); 
                \fi
                    
            \else    
                    \draw (b\i) -- (b\j); 
            \fi
        \fi
    }
}
\end{tikzpicture}
    \caption{Visualization of the profile $A$ for the proof of \Cref{thm:iterativeNoShow}.
    A black, orange, or blue edge between two alternatives means that there is (are) exactly one, two, or three voter(s) who approve(s) the connected candidates, respectively.}
    \label{fig:AlmostCompleteGraph}
\end{figure}

We next determine the winning committees for $A$ and first note that $g(A,k,\emptyset)=C$ by standardness. Hence, $(a_1)$ is a valid sequence. 
Our first goal is to show that all ways of extending $(a_1)$ lead to the committee $\{a_1,b_1,\dots, b_r\}$.
First, concurrence implies that $g(A,(a_1))= \{b_1\}$, 
as all other candidates share some ballot with $a_1$.
Hence, the only valid continuation is $(a_1,b_1)$. 
We next suppose inductively that $(a_1,b_1,\dots, b_{\ell-1})$ is a valid sequence with $\ell<r$. 
We now check the requirements for concurrence and consider to this end the candidate $b_\ell$. First, we note that there is one voter who reports $\{b_\ell, x\}$ for each alternative $x\in \{a_1,b_1,\dots, b_{\ell-1}\}$. 
By contrast, for each candidate $a_i$ with $i\geq 2$, there is at least one voter who reports $\{a_i,x\}$ for each $x\in \{a_1, b_2,\dots, b_{\ell-1}\}$ and three voters who report $\{b_1,a_i\}$. 
So, concurrence implies that $a_i\not\in g(A,k,(a_1,b_1,\dots,b_{\ell-1}))$. 
Similarly, we can check that $b_r$ is not in this set because two voters report $\{b_1,b_r\}$ but only a single voter reports $\{b_1,b_\ell\}$. 
Hence, $g(A,k,(a_1,b_1,\dots,b_{\ell-1}))\subseteq \{b_\ell,\dots, b_{r-1}\}$ and due to the symmetry of these candidates in $A$, we can suppose that $b_\ell\in g(A,k,(a_1,b_1,\dots,b_{\ell-1}))$. 
For the final step, we need to compare $b_r$ with $a_i$ with $i>1$. Since $b_1$ only shares two ballots with $b_r$ and three with each candidate $a_i$ for $i\geq 2$, and each other candidate $x\in \{a_1,b_2,\dots,b_{r-1}\}$ shares one ballot with $b_r$ and at least one ballot with $a_i$, concurrence necessitates $g(A,(a_1,b_1,\dots, b_{r-1}))= \{b_r\}$. 
Finally, we conclude that $\{a_1,b_1,\dots, b_{k}\}$ is the only chosen committee of size $k$ when electing $a_1$ in the first round. 
Moreover, an analogous argument shows that $g$ can only extend the sequence $(b_r)$ to the committee $\{b_r,a_1,\dots, a_r\}$. 

As the next step, we let $A'$ denote a profile which consists of four voters who report $\{a_1,\dots, a_r\}$, $\{b_1,\dots, b_r\}$, $\{a_1\}$, and $\{b_r\}$ respectively. Furthermore, let $A^* = \lambda A + A'$ for some $\lambda\in \mathbb{N}$. By standardness, we obtain that $g(A^*,\emptyset) = \{a_1, b_r\}$ regardless of the choice of $\lambda$. Next, by continuity, we can choose $\lambda$ large enough such that $g(A^*,s) \subseteq g(A,s)$ for all sequences $s\in\mathcal S(C)$. By combining this with our previous analysis, it is now easy to infer that $f(A^*,k)=\{\{a_1,b_1,\dots, b_r\}, \{b_r, a_1,\dots, a_r\}\}$. 

Finally, to show that $f$ fails participation, we consider the profile $A^*_{-i}$ where the voter $i$ with ballot $\{a_1,\dots, a_r\}$ abstains. By standardness, it follows that $g(A^*_{-i},\emptyset) = \{b_r\}$. Furthermore, for large enough $\lambda$, we get again that $g(\lambda A+A'_{-i},k,s)\subseteq g(A,k,s)$ for all $s\in \mathcal S(C)$. This means that $f(A^*_{-i},k)=\{\{b_r, a_r,\dots, a_1\} \}$, so voter $i$ can benefit by abstaining and $f$ fails participation.
\end{proof}

As a corollary of \Cref{thm:iterativeNoShow}, it follows immediately that \sphrag, \mes, and all sequential Thiele rules but \av fail participation because 
the query functions of these rules satisfy all conditions of \Cref{thm:iterativeNoShow}.

\begin{corollary}\label{cor:}
    Every sequential Thiele rule except \av, as well as {\sphrag} and 
    {\mes} fail participation. 
\end{corollary}

\paragraph{Remark 1.} The axioms of \Cref{thm:iterativeNoShow} are independent. All Thiele rules but \av satisfy all properties but standardness and \av only violates concurrence.
    Moreover, if we adapt \av to break ties whenever concurrence requires it, the resulting rule only fails continuity. 
    By contrast, \Cref{thm:iterativeNoShow} turns into a possibility if $k\leq 2$ or if requiring that a voter needs to approve all of the elected candidates after abstaining as sequential Thiele rules then satisfy participation.

\paragraph{Remark 2.}
    We note that the proof of \Cref{thm:iterativeNoShow} can be adapted to show that also reverse sequential ABC voting rules, which start with a full committee and then iteratively delete candidates, and sequential Thiele rules with increasing marginal contribution fail participation. 
    While such rules are only rarely considered in the literature, this shows that \Cref{thm:iterativeNoShow} is rather robust. 
    Moreover, Example~7 by \citet{SaFi19a} entails that also the optimizing variants of \sphrag fail participation.
    
    By contrast, there are sequential rules other than \av that satisfy participation. In particular, \citet{DoLe23a} introduce the class of ballot size weighted approval voting (BSWAV) rules.
    Since these rules are ABC scoring rules, they satisfy group participation by \Cref{thm:scoringPossibility}. However, all these rules coincide with their sequential version, thus giving a class of sequential ABC voting rules that satisfy participation. Notably, the proof of \Cref{thm:iterativeNoShow} can be modified to show that every sequential ABC scoring rule that satisfies participation and $s(x,y)>0$ for all $x,y>0$ belongs to this class. 

\subsection{Axiomatic Escape Routes to \Cref{thm:iterativeNoShow}}\label{subsec:axiomaticEscapeRoutes}

We next consider two escape routes to \Cref{thm:iterativeNoShow}: we first show that at least voters who do not approve any of the elected candidates cannot benefit by abstaining for most sequential ABC voting rules, and then that these rules satisfy participation on the important special case of laminar profiles.

\paragraph{Unrepresented Voters.}

The proof of \Cref{thm:iterativeNoShow} shows that voters who only approve a single elected candidate can significantly gain by abstaining and it is easy to extend this result to voters who approve more than one elected candidate. Hence, the only open case is whether voters who approve none of the elected candidates can benefit by abstaining. In analogy to the notion of strategyproofness for unrepresented voters by \citet{DDE+22a}, we thus require that a voter who approves none of the elected candidates cannot benefit by abstaining. More formally, we say an ABC voting rule $f$ satisfies \emph{participation for unrepresented voters} if $f(A_{-i},k)\not\succ_i f(A,k)$ for all approval profiles $A$, committee sizes $k$, and voters $i\in N_A$ for which there is a committee $W\in f(A,k)$ with $W\cap A_i=\emptyset$. 
As we show next, all sequential Thiele rules and \sphrag satisfy this condition, whereas \mes even fails this minimal notion of participation. 
We note that these results complement \Cref{thm:iterativeNoShow} by showing that the violation of participation observed in this result is maximal and, moreover, strengthen insights by \citet[][Prop.~A.3]{LaSk22b} on a mild monotonicity axiom.

\begin{restatable}{proposition}{unrep}\label{prop:unrep}
        Sequential Thiele rules and {\sphrag} satisfy participation for unrepresented voters. \mes violates participation for unrepresented voters. 
\end{restatable}
\begin{proof}[Proof sketch]
    The key insight why \sphrag and sequential Thiele rules satisfy participation for unrepresented voters is that an abstaining voter who does not approve any of the elected candidates cannot affect the picking sequence of the candidates. 
    Indeed, the scores of her approved candidates are always too low to be picked and the voter only further reduces these scores by abstaining. 
    By contrast, for \mes, we construct a counterexample by using that an abstaining voter influences the budgets of other voters.
\end{proof}

\paragraph{Laminar Profiles.}
As our second escape route to \Cref{thm:iterativeNoShow}, we consider the effect of restricting the domain of feasible profiles. In particular, we will show that sequential Thiele rules, \sphrag, and \mes satisfy participation on \emph{laminar profiles}. To this end, we say that a profile $A$ is laminar if for all candidates $c,d\in C$, it holds that $N_A(c)\subseteq N_A(d)$, $N_A(d)\subseteq N_A(c)$, or $N_A(c)\cap N_A(d)=\emptyset$. 
These profiles have been introduced by \citet{PeSk20a} who additionally require size constraints on the sets $N_A(c)$ which make no sense in our context. 
Laminar profiles generalize the concept of party-list profiles \citep[e.g.,][]{BLS18a,Bota21a} and thus constitute an important special case of approval profiles. 

\begin{restatable}{proposition}{laminar}\label{prop:laminar}
    Sequential Thiele rules, \sphrag, and \mes satisfy participation on laminar profiles.
\end{restatable}

\begin{proof}[Proof Sketch]
The central idea for this proposition is that for laminar profiles $A$, the sets $N_A(c)$ can be represented by a forest $F$ on the candidates where $c$ is a child of $d$ if $N_A(c)\subseteq N_A(d)$. 
For all considered rules, it then follows that it is always valid to add a candidate to the winning committee before any of her children.
This structure on the picking order can be used to show that a voter cannot increase the number of her approved
candidates by abstaining as the picking order imposed by the profile does not change.
\end{proof}

\subsection{Hardness of Abstention}
\label{subsec:hardness}
In this section, we investigate the following decision problem: given an election instance and a voter, can this voter benefit by abstaining from the election when a specific rule (e.g., \mes or \spav) is used.
This offers yet another perspective on \Cref{thm:iterativeNoShow}:
while many sequential rules fail participation, it can still be the case that the voters---even when knowing the preferences of all other participants---are unlikely to be able to efficiently decide whether they can benefit from abstention. 

The general proof idea for our reductions is as follows: the reduced instances consist of a part of the election mimicking an \NP-hard problem together with a gadget that is a small election where it is beneficial to abstain for a voter.
The elections then consist of essentially two stages. 
First, we select candidates associated with a proposed solution to the \NP-hard source problem, e.g., a set of vertices of a given target size.
Then, certain gadget candidates are selected.
If the proposed solution to the \NP-hard problem is not of the desired form, e.g., if the vertex set is not an independent set, then the selection of gadget candidates cannot be influenced by abstention. However, if the source instance is a Yes-instance, then some voter approving gadget candidates can benefit from abstention.

As a first result in this section, we find that participation leads to a computational intractability for sequential Thiele rules.
The proof is inspired by \citet{JaFa23b}.
We showcase the proof and the general proof technique of this section by considering the special case of \spav.

\begin{restatable}{theorem}{SeqThieleNP}\label{thm:SeqThieleNP}
    For every sequential Thiele rule except \av, it is \NP-hard to decide whether a voter can benefit from abstention.
\end{restatable}

\begin{proof}[Proof sketch for \spav]
	We perform a reduction from \textsc{IndependentSet} for cubic
    graphs \citep{GaJo79a}. Given an instance $(G,t)$ of \textsc{IndependentSet}, where $G = (V,E)$ is a cubic graph and $t\in\mathbb{N}$ is the target size of the independent set, we construct the following reduced instance.
	Without loss of generality, we assume that we only consider instances where $|V| \ge 3$, $|E|>0$, $|V|^3$ is divisible by~$8$, and $t$ is divisible by~$2$.
    Since the independent set needs to be of size $t$ and the graph is cubic, we assume $|E|\ge 3t$.

	The set of candidates is $C = \{g_i\colon i\in \{1,\dots,4\} \}\cup C_V$, where
    $C_V = \{c_v\colon v\in V\}$.
	The candidates $g_i$, called \emph{gadget candidates}, form a gadget in which abstention might be performed and candidates $c_v$, called \emph{vertex candidates}, represent vertices $v$ of the source instance.

Let $x = |V|$ and $y = \frac 12\left(x^4 -t\frac {x^3}2 +\frac {x^3}4\right)$. This is an integer by assumption.
The approval profile is given as follows:
	\begin{itemize}
		\item For each vertex $v\in V$, there exist $x^3$ voters approving $\{c_v\}$. For each pair of vertices $\{v,w\}\subseteq V$, there exist $x^3$ voters with approval set $\{c_v,c_w\}$. For each edge $\{v,w\}\in E$, there exists one voter with approval set $\{c_v,c_w,g_1\}$.
		\item Moreover, there are voters approving only the gadget candidates. These are  $y$ voters for each of the approval sets $\{g_1,g_2\},\{g_1,g_3\},\{g_2,g_4\},\{g_3,g_4\}$, $|E|-\frac 32t\in \mathbb N_{>0}$ voters approving $\{g_2\}$ 
        (recall that $t$ is divisible by $2$, $|E|>0$, and $|E| \ge 3t$), and  one voter approving $\{g_1\}$.
	\end{itemize}

    As usual, $A$ denotes the approval profile.
	The target committee size is $k=x +3$. 
	Eventually, we will select all candidates in $C_V$ as well as $3$ gadget candidates.
	We claim that 
		a voter with approval set $\{g_1,g_3\}$ can benefit from abstention if and only if the source instance is a Yes-instance.

    We qualitatively describe the election by {\spav} in the reduced instance.
    Initially, the score of vertex candidates is larger than the score of gadget candidates.
    By design of the scores, the first $t$ candidates to be selected are vertex candidates.
    Now, the marginal gain to the score by gadget candidates overtakes the gain by vertex candidates, and we select two gadget candidates.
    The reduced instance is designed in a way such that without abstention, $g_1$ always has the highest score among gadget candidates and is selected first. Then, the candidate $g_4$ is selected.
    Afterwards, the committee is filled with the remaining vertex candidates and then a third gadget candidate.
    Since $g_2$ contributes more than $g_3$, we select $g_2$ as the final candidate in the committee.
    Together, the choice set contains exactly the committee $C_V\cup\{g_1,g_2,g_4\}$.

        If, however, a voter with approval set $\{g_1,g_3\}$ abstains from the election, the election is similar but the committee $C_V\cup\{g_1,g_2,g_3\}$ may be selected additionally if the source instance was a Yes-instance.
		
		In summary, if the source instance is a No-instance, then the choice set is identical after abstention, and there is no incentive to abstain.
		Otherwise, the choice set additionally contains $C_V\cup\{g_1,g_2,g_3\}$ and is preferred by a voter with approval set $\{g_1,g_3\}$ (due to Kelly's extension).
\end{proof}

As a next result, we want to consider the method of equal shares.
Importantly, an execution of {\mes} may heavily rely on the completion method applied in Phase~2.
We thus provide a reduction where a voter benefits from abstention both after Phase~1 and Phase~2 of {\mes}.
The same reduction also yields a result for {\sphrag}.

\begin{restatable}{theorem}{hardMES}\label{thm:hardMES}
    Consider voting by {\mes}.
    Then, 
    \begin{enumerate}
        \item it is \NP-hard to decide whether a voter can benefit from abstention after Phase~1.
        \item it is \NP-hard to decide whether a voter can benefit from abstention after Phase~2, 
        even if none of her approved candidates are elected in Phase~2.
    \end{enumerate}
    
    Moreover, for {\sphrag}, it is \NP-hard to decide whether a voter can benefit from abstention.
\end{restatable}

The proofs of our hardness results are quite universal and allow for a number of interesting consequences. 
First, if we omit the abstaining voter from the reduced instance, then there is exactly one possibly selected committee if the source instance is a No-instance, and two possible committees if the source instance is a Yes-instance. 
Hence, we recover the hardness of deciding whether the election contains more than one possible committee \citep{JaFa23b}.

Second, it is interesting to see why we formulate our theorems as \NP-hardness, but not \NP-completeness, i.e., we do not know whether membership in {\NP} holds.
In fact, it is unclear whether this actually is true. 
In particular, we cannot use the outcomes of elections as polynomial-size certificates for verifying whether a voter benefits from abstention because we cannot check this in polynomial time.

\begin{restatable}{corollary}{Verification}\label{cor:verification}
For every sequential Thiele rule except {\av}, as well as for Phase~1 or complete {\mes}, and {\sphrag}, the following statements are true.
    \begin{enumerate}
        \item Given a set of committees $\mathcal C$, it is \coNP-complete to decide whether $\mathcal C$ is the outcome of the election.
        \item Given a positive integer $s$, it is \NP-complete to decide whether a given voter approves at least $s$ candidates in some winning committee.
    \end{enumerate}
\end{restatable}

Finally, our reductions give novel insights into the robustness of sequential ABC voting rules.
\citet{FGK22b} consider the question whether the outcome of an election can change if a given number of approvals of candidates can be added or removed. They find that this problem is \NP-hard for {\sccav}, {\spav}, and {\sphrag}.
However, they operate in a setting, where the election of a single committee is enforced by lexicographic tie-breaking and they need an unbounded budget of approvals to be added or deleted (but of linear magnitude with respect to the size of the source instance).
As a third corollary from our reductions, we complement their results in the set-valued setting and obtain hardness even if we only add or delete a single approval. The result holds because for a tied second committee, only the approval of $g_1$ by the abstaining voter is relevant. So, if the source instance is a Yes-instance, then this approval can be added or deleted to create or prevent a second outcome of the election. The reductions can be modified so that the addition or deletion of other candidates does not matter for its outcome.

\begin{corollary}
    For every sequential Thiele rule except {\av}, as well as for Phase~1 or complete {\mes}, and {\sphrag}, it is \NP-complete to decide if the outcome of the election can change if a single approval is added or deleted. 
\end{corollary}

\section{Conclusion}

In this paper, we initiate the study of participation for ABC voting rules. This axiom describes that it can never be beneficial for voters to abstain and thus describes an incentive for active participation in elections. 
In more detail, we prove that all ABC scoring rules even satisfy group participation, thus generalizing a prominent phenomenon from single-winner voting to ABC elections. Moreover, we give a strong impossibility theorem demonstrating that most sequential rules (e.g., all sequential Thiele rules but \av, sequential Phragm\'en, and the method of equal shares) severely fail participation. 

In light of this strong negative result for sequential rules, we then explore various escape routes. In particular, we show that sequential Thiele rules and sequential Phragm\'en satisfy participation for voters who do not approve any candidate in the winning committee as well as participation on laminar profiles. These results demonstrate that sequential rules satisfy participation at least in important special cases. 
Finally, we also show that for all commonly studied sequential ABC voting rules, it is \NP-hard to decide for a voter whether she can benefit by abstaining. 
This indicates that, while voters can in general benefit by abstaining, they may not be able to recognize this. 
Our approach for deriving these results is very universal and allows us to recover, strengthen, and extend existing hardness results.

\section*{Acknowledgements}
We thank Felix Brandt, the participants of the Summer School on Computational Social Choice (University of Amsterdam, 2023), and the anonymous reviewers for helpful feedback.
This work was supported by the Deutsche Forschungsgemeinschaft under grant BR 2312/12-1 and by the AI Programme of The Alan Turing Institute.

\newpage
\appendix
\section{Notation and Definitions}\label{app:notation}
To facilitate the proofs for sequential Thiele rules, we introduce the following notation for the \emph{marginal scores} of the candidates. 
    Let $s$ be the scoring function of a sequential Thiele rule.
    Given a candidate $c$, an approval profile $A$, and a partial committee $P$, we define the \emph{marginal score} of $c$ with respect to $A$ and $P$ as 
    $$\sco[A]{c}{P} = \sum_{i'\in N_A} s(|A_{i'}\cap(P\cup\{c\})|)-s(|A_{i'}\cap P|)\text.$$
Note that maximizing the Thiele score of the committee $P\cup\{c\}$ for $c\in C\setminus P$
is equivalent to maximizing the marginal score, i.e., $\sco[A]{c}{P}$.

For \sphrag, we use an alternative definition based on the idea of load balancing that can be found in the book by \citet{LaSk22b}:
This is based on variables $y_r(i)$ for a given round $r$, denoting how many candidates have been elected so far, and voter $i$. 
Intuitively, $y_r(i)$ denotes the cost that voter $i$ contributed to the committee after the first $r$ members have been elected.
We start with $y_0(i)= 0$ for all $i\in N_A$. 
When fixing any not yet elected $c\in C$ and having already elected $r-1$ members $(c_1,\dots, c_{r-1})$, we now look at the maximum cost that some voter will have to contribute to elect $c$. We have to take into consideration everything the voters approving $c$, i.e., $N_A(c)$, have paid before and that the cost of candidate $c$ is one. The best way to split the costs such that the maximal total contribution will be minimized leads to each voter paying exactly $\ell_r(c) = \frac{1+\sum_{i_\in N_A(c)}{y_{r-1}(i)}}{|N_A(c)|}$. Here, we omit for readability the sequence $(c_1,\dots, c_{r-1})$ in the argument of $\ell_r$.
Hence, already having elected $(c_1, \dots, c_{r-1})$, {\sphrag} adds one of the candidates $c\notin \{c_1\dots c_r\}$ minimizing $\ell_r(c)$. Having added this candidate, say $c_r$, the loads that the voters have contributed so far are updated. We have $y_r(i)=y_{r-1}(i)$ for all $i\notin N_A(c_r)$, and $y_r(i)=\ell_r(c_r)$ for all $i\in N_A(c_r)$. 

\section{Proof of \Cref{cor:}}
We prove the corollary by showing that the considered rules satisfy continuity, standardness, and concurrence. 
Then, the assertion follows from \Cref{thm:iterativeNoShow}.
For continuity, we will use the following observation:

Let $A^*= \lambda A + A'$. We need to show for the rules that $g(A^*,S)\subseteq g(A,S)$ for large enough $\lambda$.
Hence, if $g(A,S)$ chooses the maximizers or minimizers of some function $v$, and $g(A^*,S)$ does the same with respect to $v_\lambda$, it suffices to show that $v_\lambda$ converges to $v$ pointwise (for $\lambda$ tending to infinity).
For a sequence of functions $h_\lambda$ and an additional function $h$, we write $h_\lambda \to h$ to denote that $h_\lambda$ converges to $h$ pointwise.

\subsection{Sequential Thiele rules.}
Consider a sequential Thiele rule given by a scoring function~$s$.
We first assume that 
$s(2)-s(1)< s(1)$.
\paragraph{Standardness}
For all $c\in C$, $\sco[A]{c}{\emptyset}= \sum_{i'\in N_A} s(|A_{i'}\cap(\{c\})|)-s(0) = |N_A(c)| s(1)$, proving the claim.
\paragraph{Continuity}
It suffices to show $\frac{\sco[A^*]{c}{P}}{\lambda} \to \sco[A]{c}{P}$.
Let $P\subset \mathcal{C}$ be a partial committee.
Then, $\sco[\lambda A+ A']{c}{P} =\lambda  \sco[A]{c}{P}+ \sco[A']{c}{P}$.
Clearly $\sco[A]{c}{P}+\frac{\sco[A']{c}{P}}{\lambda} \to  \sco[A]{c}{P}$. 
\paragraph{Concurrence}
Let $A$ and $c,d, (c_1,\dots c_\ell)$ be given as required.
Then, the marginal score of $c$ decreases equally for each ballot of the form $\{c_j,c\}$, analogously for $d$ with $\{c_j,d\}$. Since $d$ shares more such ballots, it has a smaller marginal score and can thus not be chosen.
\medskip

If we don't have $s(2) - s(1) < s(1)$, then, as $s$ does not represent $\av$, we have that $s(2+x)-s(1+x)< s(1+x) -s(x)$ for some minimal $x>0$.
Hence, we can simply modify the profile $A$ from the proof of \Cref{thm:iterativeNoShow}.
By adding $d_1\dots, d_x$ to each ballot in $A$, we know that the sequential Thiele rule must first choose all $d_j$ and then proceeds as in the previous case.

\subsection{Sequential Phragm\'en.}

Next, we consider {\sphrag}.

\paragraph{Standardness.}
It holds that $\ell_1(c) = \frac{1+\sum_{i_\in N(c)}{0}}{|N_A(c)|}$, which is minimized by the approval winners.

\paragraph{Continuity.}
For a profile $A$, let $y_r$ and $\ell_r$ be defined as above, where $y_r(i')=0$ if $i'\notin N_A$. For the profile $\lambda A+ A'$, we use $y_r^*$, $\ell_r^*$. Note that $y_r^*(i) = y_r^*(i')$ for all $i\in N_A$ and $i'$ that are clones of $i$. For readability, we will now write $N, N'$ instead of $N_A, N_{A'}$. 

Since we choose candidates with minimal load, it suffices to show $\lambda \ell_r^*\to \ell_r$ for establishing continuity.
For this, we show that $\lambda y_{r}^*(i)$ converges to $y_r(i)$ for all $i\in N\cup N'$.
We prove this by induction over $r\ge 0$.

For $r=0$, observe that it does not matter how we define $\ell_0$ and we can simply set $\ell_0(c)= 0= \ell^*_0(c)$ for all $c$. Further, $y_0(i)= 0= y_0^*(i)$ for all $i\in A^*$.

Assume the claim holds for $r-1$. 
Let us consider $\lambda\ell^*_r$. By splitting the sum in the numerator and dividing through $\lambda$ in the denominator, we obtain $\lambda\ell^*_r(c)=\frac{1+\sum_{i\in \lambda N(c)}{y^*_{r-1}(i)}+\sum_{i\in N'(c)}{y^*_{r-1}(i)}}{|N_A(c)|+ \frac{|N_{A'}(c)|}{\lambda}}$.
For the numerator we have $\sum_{i\in \lambda N(c)}{y^*_{r-1}(i)} = \sum_{i\in N(c)}{\lambda y^*_{r-1}(i)}$ which converges to $y_{r-1}(i)$ by induction. 
Similarly, $\sum_{i\in \lambda N'(c)}{y^*_{r-1}(i)}$ converges to $0$ by induction. 
The denominator converges to $|N(c)|$.
Hence $\lambda\ell^*_r(c)$ converges to $\ell_r(c)$. 
From this, it directly follows that in both cases $\lambda y_r^*(i)$ converges to $y_r (i)$. 

\paragraph{Concurrence.}
Let an approval profile $A$ and candidates $c,d$, and $(c_1,\dots c_\ell)$ be given as required. 
Further, let $t_1,\dots, t_\ell$ denote the times at which $c_1,\dots, c_\ell$ were bought for the committee. 
Then, the total budget of $c$ increases between all $t_j$ at least as fast as the total budget of $d$.
Additionally, since $c$ and $d$ are both unchosen so far, each voter approving $\{c_j, c\}$ has the same budget as each voter approving $\{c_j, d\}$ at any time $t\leq t_\ell$.
Thus, at each time $t_j$ for $j\leq \ell$, the budget of $d$ is decreased at least as much as the budget of $c$, and at some time it is decreased strictly more. 
In total, the budget of $c$ will always be strictly larger than that of $d$ from time $t_\ell$ onward, until the next candidate is chosen. Hence, $d\notin g(A,k,(c_1,\dots, c_\ell))$.

\subsection{MES.}
For \mes, we can show an even stronger statement:
For profiles $A$ where $|A_i|\leq 2$ for all $i\in N_A$ and $|N_A(c)|=|N_A(d)|\leq \frac{n}{k}$ for all $c,d\in C$, the rule is equal to \sphrag. For this, it suffices to show that after Phase~1, the budgets are the same.
For the direction from left to right, let $(c_1,\dots, c_\ell)\in \mathcal S$ be a sequence that can be bought in Phase~1.
This implies that there are no ballots of the form $\{c_j, c_{j'}\}$ with $j\neq j' \leq \ell$, as else the budget would not suffice.
Further, it must be that $|N_A(c_j)|= \frac{n}{k}$ for all $j$. The sequence is hence also valid under \sphrag, since at time $t= \frac{k}{n}$, the voters approving $c_j$ have enough money to buy this candidate into the committee for all $j\leq \ell$. For both rules, the voters approving the candidates are afterwards left with a budget of zero.

For the direction from right to left, let $(c_1,\dots, c_\ell)\in \mathcal S$ be a sequence that can be bought by \sphrag up to time $t= \frac{k}{n}$. 
Then, it must be the case that all of them have $N_A(c_j)= \frac{n}{k}$ and they share no ballots. Hence, they can also be bought in Phase~1 of \mes. Again in both cases, all voters approving some $c_\ell$ are left with a budget of zero, concluding the proof.

For a violation of participation by \mes that is severe and does not rely on Phase~2, note that the gadget from the proof of \Cref{thm:hardMES} (see \Cref{fig:MESgadget}) can be generalized such that, already after Phase~1, the abstaining agent only has one candidate in the committee when participating and all of her approved candidates if abstaining. This works for an arbitrary number of approved candidates.

\section{Proofs of \Cref{prop:unrep,prop:laminar}}
Next, we consider participation for unrepresented voters.

\unrep*
\begin{proof}
    We prove the claim here for sequential Thiele rules; the claim for \sphrag follows analogously. 
    Thus, let $f$ denote a sequential Thiele rule and let $s$ denote its scoring function. 
    Moreover, we let $A$ denote a profile, $k$ the target committee size, $i\in N_A$ a voter, and $W\in f(A,k)$ a winning committee with $W\cap A_i = \emptyset$.
    
    Our goal is to show that $f(A_{-i}, k)\not \succ_i f(A,k)$.
    For this, we show two claims, namely first that $W\in f(A_{-i},k)$ and second if there is $W'\in f(A_{-i},k)$ with $W'\cap A_i\neq\emptyset$, then there is also $W''\in f(A,k)$ with $W''\cap A_i\neq\emptyset$. 
    
    First, we show that $W\in f(A_{-i},k)$.
    To this end, let $S=(c_1,\dots, c_k)$ denote the sequence in which the candidates in $W$ get picked. The central observation is now that $\hat s(A_{-i},\{c_1,\dots, c_{\ell}\})=\hat s(A,\{c_1,\dots, c_{\ell}\})$ for all $\ell\in \{1,\dots, k\}$ since voter $i$ does not approve any of the candidates in $W$. 
    Moreover, it holds that $\hat s(A, \{c_1,\dots, c_\ell\})\geq \hat s(A, \{c_1,\dots, c_{\ell-1},x\})\geq \hat s(A_{-i}, \{c_1,\dots, c_{\ell-1},x\})$ for all $x\in C\setminus \{c_1,\dots, c_\ell\}$, where the first inequality follows from definition of sequential Thiele rules, and the second one that voter $i$ may approve $x$ and thus reduce the points of committee by abstaining. 
    Chaining our inequalities, we thus infer that $\hat s(A_{-i},\{c_1,\dots, c_{\ell}\})\geq \hat s(A_{-i}, \{c_1,\dots,c_{\ell-1},x\})$ for all $\ell \in \{1,\dots, k\}$ and $x\in C\setminus \{c_1,\dots, c_\ell\}$, so $S$ is also a chosen sequence for $A_{-i}$ and $W\in f(A_{-i}, k)$.

    For the second point, we assume that there is $W'\in f(A_{-i},k)$ such that $W'\cap A_i\neq\emptyset$ and let $S'=(c_1',\dots, c_k')$ denote the corresponding picking process. 
    Now, suppose that $(c_1',\dots, c_\ell')$ is the longest prefix of $S'$ that is also valid for $A$ (if $\ell=0$, this is the empty sequence). 
    If voter $i$ approves any candidate in this prefix, then there is a committee $W''\in f(A,k)$ with $W''\cap A_i\neq \emptyset$ and we are done. Hence, suppose that $\{c_1',\dots,c_\ell'\}\cap A_i=\emptyset$, which implies that $\ell<k$ as $W'\cap A_i\neq\emptyset$. In this case, it is easy to see that $\hat s(A,\{c_1,\dots, c_{\ell+1}\})\geq \hat s(A_{-i}, \{c_1,\dots, c_\ell,x\})$ for every $x\in C\setminus (\{c_1,\dots, c_\ell\}\cup A_i)$ as voter $i$ does not change the score of $x$ but may increase the score of $c_{\ell+1}'$. Hence, since the candidate $c_{\ell+1}'$ is not valid anymore as $(c_1',\dots, c_\ell')$ is the longest valid prefix of $S$, this means that we add next a candidate that is approved by voter $i$, so there is a committee $W''\in f(A,k)$ with $W\cap A_i\neq\emptyset$.

    We next turn to the counter example demonstrating that \mes fails participation for unrepresented voters. To this end, we consider the following profile $A$ consisting of $51$ voters and set our target committee size to $5$.\smallskip
    
    \begin{tabular}{lll}
        $20\times \{x_1,x_2,x_3\}$ & $10\times \{y_1,y_2\}$ & $10\times \{y_1,y_2,z\}$ \\
        $9\times \{z,c\}$ & $2\times \{c\}$
    \end{tabular}\smallskip

    We claim that $c\not\in W$ for every committee $W\in f(A,5)$, but $c$ will be chosen if one of the voters uniquely approving it abstains. 
    We start by analyzing the profile $A_{-i}$ where one of these voters abstains. In this case, there are exactly $50$ voters, so every voter has a budget of $\frac{1}{10}$. It is now easy to check that two of the candidates $\{x_1,x_2,x_3\}$ and both $y_1$ and $y_2$ will be bought for a price of $\frac{1}{20}$. Finally, as fifth candidate $c$ will be bought for a price of $\frac{1}{10}$. 

    By contrast, in the profile $A$, every voter has a budget of $\frac{5}{51}$, so we can only buy one candidate from $\{x_1,x_2,x_3\}$ and one from $\{y_1,y_2\}$ (again for a price of $\frac{1}{20}$). In the next step, we then buy candidate $z$ for a price of $\frac{53}{918}$. Afterwards, Phase~1 of \mes is complete as no candidate can be afforded. It can be checked that the voters approving $\{x_1,x_2,x_3\}$ have in total a budget of $20(\frac{5}{51}-\frac{1}{20})$ left, the voters who approve $\{y_1,y_2\}$ have a budget of $10(\frac{5}{51}-\frac{1}{20})$ left, and the voters approving $c$ and $\{c,z\}$ have a budget of $11\cdot \frac{5}{51}-9\cdot \frac{53}{918}=\frac{19}{34}$. 
    Since the candidates in $\{x_1,x_2,x_3\}$ and in $\{y_1,y_2\}$ are approved by $20$ voters whereas $c$ is only approved by $10$ voters, it can now be easily verified that we will buy one candidate of both of these two sets in Phase~2 of \mes instead of $c$.
    
    Hence, for~$A$, no committee containing $c$ is chosen whereas every committee contains $c$ in $A_{-i}$. 
    We note that this example also generalizes an observation by \citet{LaSk22b} which implies that \mes fails participation for unrepresented voters when using a suitable tie-breaking mechanism.
\end{proof}

Now, we consider participation on laminar profiles.

\laminar*

\begin{proof}
    The proof follows three key steps. First, we show that laminar profiles can be represented as directed forests. Second, we show how to locate valid sequences at the tops of trees. Third, we analyze changes in the score when adding candidates to the committee or a voter joins the election.
    
    We first consider sequential Thiele rules given by a scoring function $s$.
    We assume first that the scoring function is strictly increasing, i.e., that $s(x) < s(x+1)$ for all $x\ge 0$.

    \textbf{Step 1: Laminar profiles are directed forests.}
    First of all, we show that laminar profiles resemble a collection of directed trees. Even stronger, each node, up to clones, represents a ballot and has an in-degree of at most one.
    Formally, define the relation $\to$ as follows:
    Let $A$ be a laminar profile.
    $x\rightrightarrows y$ iff there is $i\in N_A$ with $x,y\in A_i$ and there is $j$ with $x\in A_j, y\notin A_j$. In other words, $x\rightrightarrows y$ iff $\emptyset\neq N_A(y) \subset N_A(x)$ This relation is asymmetric, but transitive, and we hence thin it out by defining $x\to y$ iff $x\rightrightarrows y$ and there is no $z\neq x$ such that $x\rightrightarrows z$ and $z\rightrightarrows y$.
    We claim that every node can have at most one predecessor with respect to $\to$, up to clones. 
    To show this, let $x\to z$ and $y \to z$. Then, this means that there is some $i$ with $z,x\in A_i$, which implies $x,y,z\in A_i$. Thus, the support of $x,y$ is not disjoint and one of the subset-relations must hold. Without loss of generality, let $N_A(x)\subseteq N_A(y)$. If it were that $\emptyset\neq N_A(x)\subset N_A(y)$, we could infer $y\rightrightarrows x$, which contradicts $y\to z$.
    Hence, $N_A(x)= N_A(y)$, implying that $x=y$ or $y$ is a clone of $x$.
    Hence, we can identify each $c$ that appears in the forest (including all its clones $x,y,\dots$) with the ballot including all clones and predecessors that have to appear as soon as $c$ is contained, i.e., $B(c) = \{d |N_A(d) = N_A(c) \lor d\to \dots \to c\} = \{d |N_A(c) \subseteq N_A(d)\}$. Further, we can assign to each node (=candidate) $c$ as weight the number of ballots $A_i$ that are equal to $\{d\in C| N_A(d)\subseteq N_A(c)\}$. This then fully identifies the laminar profile.

    \textbf{Step 2: Valid sequences are located at the tops of the trees.}
    We now claim the following, which will help us for the proof of participation:
    Let $(c_1,\dots, c_\ell)$ be a valid sequence w.r.t. $A$. If $c\in g(A,k,(c_1,\dots c_\ell))$, $c$ must be maximal w.r.t. $\to$ (and $\rightrightarrows$) except for $c_j$ with $j\leq \ell$.
    Let $c\to d$ and both unchosen. Since $N_A(d)\subset N_A(c)$, every voter approving $d$ must also approve $c$. Hence, for the committee $W=\{c_j | j\leq \ell\}$, we have $|A_i \cap W\cup \{c\}| = |A_i \cap W\cup \{d\}|$
    and hence $ s(|A_i \cap W\cup \{c\}|) = s(|A_i \cap W\cup \{d\}|)$ for all $i\in N_A(d)$. (For $\ell=0$, $W=\emptyset$.) Further, there must be some voter $i\in N_A(c)\setminus N_A(d)$. For this voter, $|A_i \cap (W\cup \{c\})| > |A_i \cap (W\cup \{d\})|$ and hence
    $s(|A_i \cap (W\cup \{c\})|)> s(|A_i \cap (W\cup \{c\})|)$. For all other voters $i'$, $|A_{i'} \cap (W\cup \{c\})| \geq |A_{i'} \cap (W\cup \{d\})| = |A_{i'} \cap W|$ 
    We conclude $\hat s(A, W\cup\{c\})> \hat s(A, W\cup\{d\})$ and $d\notin g(A,k,(c_1,\dots, c_\ell))$, as desired.
    
    \textbf{Step 3: Changes in score when $A_i$ joins or $a\in A_i$ is added to the sequence.}
    Let now $P\subset C$, $A$ be any laminar profile, and $a,b,c\in C\setminus P$.
    We further require that $P$ is located at the top of the trees, i.e., there is no candidate in $C \setminus P$ that dominates any candidate in $P$ w.r.t. $\rightrightarrows$. Further, we require that $a$ is located directly beneath $P$, i.e., if $a$ has a predecessor w.r.t. $\to$, it is contained in $P$.
    We now investigate how the scores are affected by adding $a$ to $P$.

    Let $c$ be from a different sub-tree than $a$, i.e., $N_A(c)\cap N_A(a) = \emptyset$. Then,
    the marginal contribution of $c$ remains the same when adding $a$ to $P$:
    $\sco[A]{c}{P\cup\{a\}} 
    =\sum_{i'\in N} s(|A_{i'}\cap(P\cup\{a,c\})|)-s(|A_{i'}\cap (P\cup\{a\})|)
    =0 + \sum_{i'\in N_A(c)} s(|A_{i'}\cap(P\cup\{a,c\})|)-s(|A_{i'}\cap P\cup\{a\}|) = 0 + \sum_{i'\in N_A(c)} s(|A_{i'}\cap(P\cup\{c\})|)-s(|A_{i'}\cap P|)=
    \sum_{i'\in N} s(|A_{i'}\cap(P\cup\{c\})|)-s(|A_{i'}\cap P|)
    = \sco[A]{c}{P}\text.
    $
    
    Now, let $b,c$ be successors or clones of $a$, i.e., $N_A(b)\subseteq N_A(a), N_A(c)\subseteq N_A(a)$ We claim that if $b$ has a weakly larger marginal contribution than $c$ before adding $a$, the same holds true after adding $a$.
    
    After adding $a$, we have
    \begin{align*}
        &\sco[A]{c}{P\cup\{a\}} - \sco[A]{b}{P\cup \{a\}} \\
        =&\sum_{i'\in N_A} s(|A_{i'}\cap(P\cup\{a,c\})|)- s(|A_{i'}\cap(P\cup\{a,b\})|) \\
        &+ s(|A_{i'}\cap (P\cup\{a\})|)-s(|A_{i'}\cap (P\cup\{a\})|)\\     
        =&\sum_{i'\in N_{A}(b)\cup N_A(c)} s(|A_{i'}\cap(P\cup\{a,c\})|)- \\
          &\qquad \qquad \quad\quad s(|A_{i'}\cap(P\cup\{a,b\})|)\\
        =&\sum_{i'\in N_A{(b)}\cup N_A(c)} s(|A_{i'}\cap\{c\}|+x)- s(|A_{i'}\cap\{b\}|+x)\\
        =&|N_A(c)\setminus N_A(b)| (s(x +1)-s(x)) \\
         -&|N_A(b)\setminus N_A(c)| (s(x+1)- s(x))\text.
    \end{align*}
    Here, the second to last line holds for some $x$ we now specify. This is because all ballots containing $b$ or $c$ must contain all their respective predecessors, including $a$, all clones of $a$ and its predecessors. Further, since the sequence constructing $P$ and adding $a$ is located at the top of the trees, $P\cup \{a\}$ must consist of all predecessors of $a$ and some clones of $a$ (including $a$ itself). We thus denote them by $P_a\subseteq P$. Let $x = |P_a|$. Then, for all $i'\in N_A(b)\cup N_A(c)$, we have that 
    $A_{i'}\cap (P\cup\{a\}) = P_a $.

    Before adding, we obtain with the same reasoning
    \begin{align*}
        &\sco[A]{c}{P} - \sco[A]{b}{P} \\
        =&|N_A(c)\setminus N_A(b)| (s(x)-s(x-1)) \\
         -&|N_A(b)\setminus N_A(c)| (s(x)- s(x-1))\text.
    \end{align*}
    This proves the claim.
    
    Now, consider the addition of some voter $i$ to $A_{-i}$. 
    For $c$ with $c\notin A_i$, the marginal score cannot change as we have
    $A_{i}\cap(P\cup\{c\}) = A_{i}\cap P$ and hence
    $\sco[A]{c}{P}  - \sco[A_{-i}]{c}{P} =  s(|A_{i}\cap(P\cup\{c\})|)-s(|A_{i}\cap P|) = 0$.
    On the other hand, if $c\in A_i$, it must be that the marginal score increases, as $\sco[A]{c}{P}  - \sco[A_{-i}]{c}{P} =  s(|A_{i}\cap(P\cup\{c\})|)-s(|A_{i}\cap P|) = s(|A_{i}\cap P|+1)-s(|A_{i}\cap P|) \geq 0$.

    \textbf{Sequence comparisons.}
    We now use Step 3 to show that the valid sequences in essence do not change when $i$ joins the election, except for the fact that there can be more candidates $c\in A_i$ appearing, and the appearance can be earlier. 

    This is due to the following: For the first candidate of the sequence, this is clear since we use 
    \texttt{AV} and the only change so far is that $A_i$ is added to $A_{-i}$. For the induction step, let the sequence $\vec c = (c_1,\dots, c_\ell)$ be given that is valid for $A_{-i}$, and let $\vec d=(d_1,\dots,d_\ell)$ be valid for $A$ and obtained from $\vec c$ such that:
    \begin{itemize}
        \item All $c_j\in A_i$ appear and are in the same order in $\vec d$ as in $\vec c$.
        \item Further $c\in A_i$ with $c\neq c_j$ for all $j\leq \ell$ can occur in $\vec d$, but only after the ones that already were in $\vec c$.
        \item All other $d_j$ occurring in $\vec d$ must satisfy $d\notin A_i$. They must appear in the same order as in $\vec c$, and with no $c_j\notin A_i$ left out. 
    \end{itemize}

    To give an example, let $\vec c= (c_1,c_2,c_3,a_1,a_2)$ and $A_i= \{a_1,a_2, a_3\}$.

    Then, $\vec d= (a_1, c_1,c_2,a_2,c_3)$, as well as $\vec d= (a_1, a_2,c_1,c_2,a_3)$ would satisfy all conditions imposed.
    On the other hand, $\vec d= (a_1, c_2,c_1,a_2,c_3)$, 
    $\vec d= (a_1, c_1,c_2,a_3,a_2)$, and $\vec d= (a_1, c_1,a_2,a_3,c_3)$ would violate some condition.

    For the induction start, note that $\ell =0 $ is trivial, as for $\vec c$ of length $0$ there is some $\vec d$ of the same length satisfying all conditions, namely the empty sequence $\vec d = \emptyset$. 

    For $\ell\to \ell+1$, let $\vec c^*$ of length $\ell+1$ be given.
    Without loss of generality, we can assume that our sequence chooses candidates of $A_i$ whenever it can. This is because we first choose all clones before we choose their successors. Hence, in each valid sequence we can permute the clones and the sequence still remains valid. If $A_i$ is the only ballot of its type in $A$, it can be that it distinguishes some candidates contained in $A_i$ from others not contained in $A_i$. Choosing the ``clones'' with respect to $A_{-i}$ that are contained in $A_i$ hence always leads to an outcome that voter $i$ weakly prefers. 
    We can use the induction hypothesis on the first $\ell$ entries of $\vec c^*$ to obtain $\vec d$ of length $\ell$.
    If we can extend $\vec d$ with some candidate from $A_i$, we are done as we can choose some $c\in A_i$ that satisfies all conditions. 
    If $\vec d$ and $\vec c$ contain exactly the same elements, we are also done as we can simply extend $\vec d$ with $c_{\ell+1}$ or some candidate in $A_i$.
    Else, let $\vec d$ contain some candidate $d_x$ be the first candidate not occurring in $\vec c$ (then, $d_x\in A_i$) and consider 
    the first $c_j$ that is not present in $\vec d$ (then, $c_j\notin A_i$). We claim that we can extend $\vec d$ with $c_j$.
    To see this, note that it is valid to extend $(c_1,\dots, c_{j-1})$ with $c_j$ with respect to $A_{-i}$. To obtain $\vec d$ from $(c_1,\dots, c_{j-1})$, we only add candidates approved in $A_i$, call them $a_1,\dots$.
    Assume for contradiction that $c_j$ is a successor of one of the $a_y$ with respect to $A$. Then, we could not have chosen $c_j$ before $a_i$ in $A_{-i}$, a contradiction. (This holds clearly in the case that $c_j$ is a successor in $A_{-i}$ too. Else, it must be a clone of some $a\in A_i$ with respect to $A_{-i}$. But then, we would prefer to choose $a$ over $c_j$, so this cannot be the case either.) 
    Since $c_j$ is not a successor to any of the $a_y$, its marginal score remains unchanged by their addition, while all other marginal scores can only decrease. Finally, when adding $i$, only the scores of candidates that $i$ approves of increase whil all others remain the same. Since the former are not chosen by assumption, $c_j$ still has maximal score and is chosen.
    This proves that for every committee in $f(A_{-i},k)$, there is a committee in $f(A,k)$ that voter $i$ prefers at least as much. Analogously, we can show that for every committee in $f(A,k)$ there is some committee in $f(A_{-i},k)$ that $i$ prefers less, concluding the proof that Thiele rules with $s(1)< s(2)< \dots$ satisfy participation on laminar profiles.

    Now consider a general Thiele rule where the scoring function is only weakly increasing.
    Then, if there exists $x\ge 1$ with $s(x)=s(x+1)$, by definition of Thiele scoring functions, $s(x')=s(x'+1)$ for all $x'\ge x$.
    We can now proceed with the same proof by induction, where the start remains trivial. To show that for each $W\in f(A_{-i},k)$ there is some $W'\in f(A,k)$ such that $|W'\cap A_i|\geq |W\cap A_i|$, we can simply take a valid partial sequence of length $\ell-1$ and substitute it with a sequence that is located at the top of the trees. More precisely, as soon as we can choose a candidate that is not located as high as possible on the tree, the marginal contribution of this candidate (and hence all others) is zero, and we can hence choose everything and extend the sequence as we wish.
    To show that for each $W'\in f(A,k)$ there is some $W\in f(A_{-i},k)$ such that $|W'\cap A_i|\geq |W\cap A_i|$, the induction start is again trivial. In the induction step we can again make a case distinction: If so far the sequence is located at the top of the tree, we can continue as before. Else, we chose a candidate with marginal contribution $0$ w.r.t. $A$. Hence, all remaining candidates from now on have marginal contribution $0$, which does not change if $A_i$ abstains. This concludes the proof for sequential Thiele rules.

    For \sphrag, the proof is analogous to the tie-less case of sequential Thiele rules.
    \textbf{Step 2} works, since $N_d\subset N_c$ implies that the total budget for $c$ will be strictly larger than the total budget of $d$ at any time $t>0$ where both are not bought yet.

    \textbf{Step 3} is substituted by the following:
    Let $A$ be laminar. 
    Instead of $\ell_r(c)$,  given a valid sequence $(c_1,\dots,c_\ell)$ for a profile $A$ we now more precisely write $\ell(A,(c_1,\dots,c_\ell), c)$  for the time at which $c$ would be bought if it was joining the committee next.
    Since the query function of \sphrag does not depend on $k$, we can omit it.
    Then, let $\vec d = (d_1,\dots,d_{\ell'})$ and $\vec c = (c_1,\dots,c_{\ell})$ be given with the same requirements as for sequential Thiele rules.

    To prove this claim, we use induction:
    $\ell= 0$ is trivial.
    $\ell\to \ell+1$, we obtain that  $\ell(A,(d_1,\dots,d_{\ell'}), d) = \ell(A_{-i},(c_1,\dots,c_\ell), d)$. The last $d_{\ell' +1}$ must also be equal to some $c_{j_{\ell'+1}}$. Let $d\notin A_i$. The $c_j$ in-between do not change the $y(i)$ of any $i\in N_A(d)$ by definition. Hence, also $\ell_r(c)$ remains the same, proving the claim.

    From this, we can directly infer the following: For $c\in A_i$, $\ell_(A,(c_1,\dots,c_\ell),c)$ will only decrease. Hence, we obtain sequences of the form $(c_1,\dots)$ when $i$ abstains and sequences of the form $(d_1,\dots)$ when $i$ participates with the same conditions imposed as for the proof of sequnetial Thiele rules.
    
    Thus, for every committee $W$ that can be chosen when $i$ abstains, there is another committee $W'$ that is chosen when $i$ participates and is weakly preferred by $i$. Vice versa, removing the ballot $A_i$ only increases the $\ell_r(c)$ of $c\in A_i$ while leaving all others intact. Hence, for every committee $W' \in f(A,k)$ there is another committee $W\in f(A_{-i},k)$ such that $i$ weakly prefers $W'$ to $W$.
    This concludes the proof for \sphrag.

For \mes, we again show that the valid sequences stay the same except for candidates $c\in A_i$ which can now appear earlier and more frequent in the sequences. To formalize this, we need to take into account both phases.
Note that due to Step~2, the order of the sequence in Phase~1 is fixed up to permutation of the clones.
For  $(c_1,\dots,c_\ell)$ that is valid and concludes Phase~1 , let $b(c,(c_1,\dots,c_\ell)) = b(c,(c_1,\dots,c_\ell,\dots))$ denote the sum of the starting budgets of the voters $i\in N_A(c)$ for  Phase~2.
Analogously to \sphrag, for  Phase~2, we do not denote the times as $\ell_r(c)$. Instead, given a valid sequence $(c_1,\dots,c_\ell)$ for a profile $A$ we now write $\ell(A,k,(c_1,\dots,c_\ell), c)$  for the time at which $c$ would be bought if it was joining the committee next. Also, the numerator replaces the $1$ with $1-b(c, (c_1,\dots,c_\ell))$ since the voters already have the starting budget.

The following statement will be helpful: Let  $(c_1,\dots,c_\ell)$ be valid and conclude Phase~1 (for $A$ or $A'$). For $c\notin \{c_1,\dots,c_\ell\}$, all voters $i\in N_{A}(c)$ ($i\in N_{A'}(c)$) have the same remaining budget.
Note that the ballots of these voters $i$ are represented precisely by the clones and successors of $c$ in the forest.
The proof by induction proceeds over valid sequences that are queried in Phase~1 of length $\ell$. For the empty sequence, it holds true. Now, for $\ell\to \ell+1$, adding a candidate $c$ in Phase~1 that apart from $\{c_1,\dots,c_\ell\}$ is at the top of some tree, the budgets of any $d$ from another tree does not change. The same holds for $d$ that are in the same tree as $c$, but not clones or successors of $c$. 
In the tree of $c$, note that by induction all voters $i$ that are represented by $c$, its clones, or its successors have the same budget left. Hence, the optimal way to balance the load always leads to them paying the same amount, concluding the induction step.

From this, we can in fact deduce that once we bought a node and all it's clones of some tree into the committee, the next chosen successor from this node (if there will be any at all) will be exactly the one with the highest approval score. As soon as we cannot afford a root or one of its clones, we cannot afford any other candidate from this tree in Phase~1.

Again, impose restrictions on $(d_1,\dots,d_{\ell'})$ and $(c_1,\dots,c_{\ell})$ which are valid for $A, A_{-i}$ respectively.
Then $\ell(A,(d_1,\dots,d_{\ell'}), d) = \ell(A_{-i},(c_1,\dots,c_\ell), d)$ for all $d\notin A_i$. 
The challenge is to manage the transition between the two phases.
$\ell= 0$ is trivial.
Let now $\ell>0$ and $(c_1,\dots,c_{\ell})$ be given.
For Phase~1, observe that voters with disjoint ballots from $A_i$ (i.e., belonging to other trees) cannot change the budget of voters with ballots not disjoint from $A_i$ by spending money. Hence, every tree buys the candidates on its own. Further, once the root and all its clones of a tree are bought and the budgets are updated, for the rest of the process we can remove the former and consider its successors to be roots. 
Also, observe that $k$ candidates can only be bought if the voters spend all their money perfectly. Hence, in this phase the different sub-trees do not compete for who gets chosen. We thus do not need to analyze thresholds between different (sub-)trees and look at each (sub-)tree individually.
Phase~1 in $A_{-i}$ was executed with a budget per voter of $\frac{k}{n-1}$.
Hence, the voters with ballots not disjoint from $A_i$, denote their number by $n_T$, had a joint budget of $(n_T-1) \frac{k}{n-1}$ to buy candidates. When $i$ joins, the budget changes to $(n_T) \frac{k}{n}\geq (n_T-1) \frac{k}{n-1}$. Hence, the voters can only buy more candidates from $A_i$ in total. This remains true after the root $a\in A_i$ is bought and we analyze the sub-tree where the root is again contained in $A_i$.
Hence, Phase~1 yields weakly more elements from $A_i$ if $i$ participates.
Further, observe that the total budget of all voters strictly decreases. Hence, for all other trees, it holds that they cannot buy more candidates in the Phase~1 when $i$ participates.
In Phase~2, these trees lose time as they first have to gather money for the candidates that they could not buy in Phase~1. Even if no such candidates exist, they still lose time as their starting budget decreased and they have to first gather money to get back to the budget that they had when $i$ abstained. Hence, also combined from Phase~1 and 2, other trees do not pose a problem for voter $i$ as they cannot elect more candidates than when $i$ abstained.
In the tree of voter $i$ itself, note that all candidates benefit from voter $i$ joining, as this gives a head start with respect to the starting budget in Phase~2: Either the candidates can now already be bought in Phase 1, or we are closer to buying them in Phase~2. Also, in Phase~2 all successors of each candidate in $A_i$ benefit from the lesser load the have to pay when one of the $A_i$ is elected: if some $c\in A_i$ needs to be bought before a successor $c\to d$, and $i$ joins, then $c$ will be bought earlier and the voters approving $d$ have more time to start from scratch and generate money for $d$.
Further, only the candidates in $A_i$ benefit from a faster money generation rate when $i$ participates. 

We can hence proceed as for sequential Thiele rules and \sphrag to obtain the result for \mes.
 \end{proof}

\section{Proofs Concerning Hardness of Abstention}

In this appendix, we provide the missing proofs from \Cref{subsec:hardness}.

We start with the complete proof of the hardness result for all Thiele rules.

\SeqThieleNP*

\begin{proof}
    Let $s\colon \mathbb N_0 \to \mathbb Q$ be any scoring function inducing a Thiele rule.
    Since we assume concavity, for $x\ge 1$, it holds that $s(x+1)-s(x) \le s(x) - s(x-1)$.
    Assume that $s$ is not the function defined by $s(x) = \alpha x$ for all $x$ and some $\alpha > 0$, i.e., represents a Thiele rule distinct from {\av}.

    For $x\ge 1$, let $\delta(x) = s(x) - s(x-1)$.
    Hence, concavity implies that $\delta(x+1)\le \delta(x)$ for all $x\ge 1$.
    Since $s$ does not represent {\av}, we know that $\delta(y+1) < \delta(y)$ for some $y\ge 1$.
    By concavity, this implies $\delta(y') < \delta(y)$ for all $y'\ge y+1$. 

    Also, we claim that there exists a $Y\ge y+1$ such that $\delta(Y-1) - \delta(Y) > \delta(Y) - \delta(Y+1)$.
    Indeed, if this was not the case, then, for every $\ell\ge y+1$, it would hold that $\delta(\ell) - \delta(\ell+1) \ge \delta(\ell-1)-\delta(\ell) \ge \dots \ge \delta(y) - \delta(y+1)$.
    Hence, $\delta(\ell + 1) \le \delta(\ell) - (\delta(y) - \delta(y+1))$, and inductively $\delta(\ell + 1) \le \delta(y) - \ell(\delta(y) - \delta(y+1))$.
    Since the right-hand side is negative for large $\ell$, we obtain a contradiction to the monotonicity of $s$.

    Hence, there exists $Y\ge 2$ with $\delta(Y-1) - \delta(Y) > \delta(Y) - \delta(Y+1)$.
    Without loss of generality,\footnote{Otherwise, one can enhance the reduction by the standard trick of adding $Y-2$ candidates approved by all voters \citep[see, e.g.,][]{JaFa23b}. Then, the sequential Thiele rule has to select these candidates first, and the election behaves as if starting with the $(Y-1)$st candidate with the capped weight function.} we may assume that this already happens for $Y =2$, i.e., we assume that
    
    \begin{equation}\label{eq:decay}
        \delta(1) - \delta(2) > \delta(2) - \delta(3)\text.
    \end{equation}

    Note that \Cref{eq:decay} together with concavity directly implies that $\delta(2) < \delta(1)$.

    We are ready to perform our reduction from \textsc{IndependentSet} for cubic\footnote{Cubic graphs are defined as those graphs where every vertex has degree exactly $3$.} graphs \citep{GaJo79a}.
	
	Assume that we are given an instance $(G,t)$ of \textsc{IndependentSet} where $G = (V,E)$ is a cubic graph and $t$ is an integer (target size of the independent set).
    Without loss of generality, we assume that we only consider instances where $|V|\ge 2$ and $|E|\ge 3t$. (If $|E| < 3t$, then there cannot exist an independent set of size $t$ in a cubic graph.)

    We construct the reduced instance.
    
	The set of candidates is $C = \{g_i\colon i\in [4]\}\cup\{b\}\cup C_V$, where
    $C_V = \{c_v\colon v\in V\}$.
	The role of the candidates is as follows:
	\begin{itemize}
		\item Candidates $g_i$ form a gadget in which abstention might be performed.
		\item Candidates $c_v$ represent vertices $v$.
        \item Candidate $b$ is very strong, and will always be selected first. This helps to balance scores in the gadget.
	\end{itemize}

    Let $n = |V|$ and $\alpha$ a positive integer such that $\frac {\alpha}4\frac{\delta(1)-\delta(2)}{\delta(1)}$ is an integer and $\alpha(\delta(1)-\delta(2))\ge \delta(1)$.
    Choosing such an $\alpha$ is possible because $\delta(1) > \delta(2)$ and $\delta(1)$ and $\delta(2)$ are rational numbers.
    Let $m = \frac {\alpha}2\left(n^4 -\left(t-\frac 12\right) n^3\frac{\delta(1)-\delta(2)}{\delta(1)}\right)$. 
    By the choice of $\alpha$, this is an integer.
    The voters with their approval sets are as follows:
	\begin{itemize}
		\item For each vertex $v\in V$, there exist $\alpha n^3$ voters approving $\{c_v\}$.
		\item For each pair of vertices $\{v,w\}\subseteq V$, there exist $\alpha n^3$ voters with approval set $\{c_v,c_w\}$.
		\item For each edge $\{v,w\}\in E$, there exists one voter with approval set $\{c_v,c_w,g_1\}$.
        \item There exist $\alpha n^5$ with approval set $\{b\}$.
        \item There exist $3t$ voters with approval set $\{b,g_2\}$.
		\item Moreover, there are voters approving only the gadget candidates. These are 
		\begin{itemize}
			\item $m$ voters for each of the approval sets $\{g_1,g_2\},\{g_1,g_3\},\{g_2,g_4\},\{g_3,g_4\}$,
			\item $|E|-3t$ voters approving $\{g_2\}$ (this is well-defined because $|E|\ge 3t$), and
			\item one voter approving $\{g_1\}$.
		\end{itemize}
	\end{itemize}

    As usual, $A$ denotes the approval profile.
	The target committee size is $1 + n +3$. 
	We write the target size as this sum to hint at the fact that we are selecting all candidates in $\{b\}\cup C_V$ as well as $3$ gadget candidates.
	We claim that 
	
	\begin{quote}
		a voter with approval set $\{g_1,g_3\}$ can benefit from abstention if and only if the source instance is a Yes-instance.
	\end{quote}

    We use the notation of \emph{marginal scores} of the candidates as introducted in \Cref{app:notation} and recall the definition here.
    Given a candidate $c$, an approval profile~$A$, and a partial committee $P$, we define the \emph{marginal score} of $c$ with respect to $A$ and $P$ as 
    $$\sco[A]{c}{P} = \sum_{i\in N} s(|A_i\cap(P\cup\{c\})|)-s(|A_i\cap P|)\text.$$

    We are ready to perform the election in the reduced instance.
	The initial scores of the candidates are
	\begin{itemize}
		\item $\sco{b}{\emptyset} = (\alpha n^5 +3t)\delta(1)$,
		\item $\sco{c_v}{\emptyset} = (\alpha n^4 + 3)\delta(1)$ for all $v\in V$, because $c_v$ occurs in $n-1$ pairs approved by $\alpha n^3$ voters each, $\alpha n^3$ times in a singleton approval set, and $v$ is incident to $3$ edges of $G$,
		\item $\sco{g_1}{\emptyset} = (2m + |E| + 1)\delta(1)$,
		\item $\sco{g_2}{\emptyset} = (2m + |E|)\delta(1)$, and
		\item $\sco{g_3}{\emptyset} = \sco{g_4}{\emptyset} = 2m\delta(1)$.
	\end{itemize}

    Note that $(2m + |E| + 1)\delta(1) < \alpha n^4 \delta(1) + n^2\delta(1) < \alpha n^5 \delta(1)$ for $n\ge 2$.
	Hence, we have to select $b$ as the first candidate.

	Next, we show that the best candidates to select are candidates of the type $c_v$.
	Assume that $\{b\}\cup S'$ is the tentative committee consisting of $\ell + 1$ candidates, where $S'\subseteq C_V$ with $0\le |S'| = \ell\le t-1$.
    Then, for $c_v\in C_V\setminus S'$, it holds that 
    
    \begin{align*}
        & \sco{c_v}{\{b\}\cup S'} \ge \alpha n^4\delta(1) - \ell \alpha n^3(\delta(1)-\delta(2))\\
        & \phantom{score}\ge \alpha n^4\delta(1) - (t-1)\alpha n^3(\delta(1)-\delta(2))\text.
    \end{align*}
	There, we only count the contribution to the score by voters with an approval set $\{c_v,c_w\}$ for any pair $\{v,w\}\subseteq V$ or an approval set $\{c_v\}$.
	
	However, for $i\in [4]$, it holds that
    \begin{align*}
        & \sco{g_i}{\{b\}\cup S'}  \le (2m + |E| + 1)\delta(1)\\
        & < \alpha n^4\delta(1) - (t-\frac 12)\alpha n^3(\delta(1)-\delta(2)) + n^2\delta(1)\\
        & = \alpha n^4\delta(1) - (t-1)\alpha n^3(\delta(1)-\delta(2))\\& \phantom{=}- \frac 12\alpha n^3(\delta(1)-\delta(2))+ n^2\delta(1)\\
        &\le \alpha n^4\delta(1) - (t-1)\alpha n^3(\delta(1)-\delta(2))\\& \phantom{=}- \frac 12 n^3\delta(1)+ n^2\delta(1)\\
        & < \alpha n^4\delta(1) - (t-1)\alpha n^3(\delta(1)-\delta(2))\text.
    \end{align*}
    In the last weak inequality, we use that $\alpha(\delta(1)-\delta(2))\ge \delta(1)$, and in the final, strict inequality that $n\ge 2$.
	
	Together, this shows that the $t$ candidates that are selected after $b$ have to be of type $c_v$.
	Let $S\subseteq C_V$ be the set of these $t$ candidates.
	
	At this point, for all $c_v\in C_V\setminus S$, $\sco{c_v}{\{b\}\cup S} \le \alpha n^4\delta(1) - t\alpha n^3(\delta(1)-\delta(2)) + 3\delta(1)$ (bounding with the case where none of the three candidates corresponding to neighbors of $v$ in $G$ are in $S$), whereas the gadget candidates have scores
	
	\begin{itemize}
		\item $\sco{g_3}{\{b\}\cup S} = \sco{g_4}{\{b\}\cup S} = 2m\delta(1)$,
		\item $\sco{g_2}{\{b\}\cup S} = (2m + |E|-3t)\delta(1) + 3t\delta(2)$,
		\item $\sco{g_1}{\{b\}\cup S} = (2m + |E|-3t)\delta(1) + 3t\delta(2) + 1$ if $S$ is an independent set, and 
		\item $\sco{g_1}{\{b\}\cup S} \ge (2m + |E|-3t)\delta(1) + 3t\delta(2) + 1 + \left[(\delta(1) - \delta(2)) - (\delta(2)-\delta(3))\right]$ if $S$ is not an independent set.
	\end{itemize}
    
	For the score of $g_2$, note that the voters with approval set $\{b,g_1\}$ only contribute $\delta(2)$.
    For the score of $g_1$ when $S$ is an independent set, since $G$ is cubic, the score is $\delta(2)$ instead of $\delta(1)$ for exactly $3t$ candidates with approval set $\{c_v,c_w,g_1\}$.
    If $S$ is not an independent set, then the score is lowered from $\delta(1)$ to $\delta(2)$ for at most $3t-1$ such voters, whereas it is lowered to $\delta(3)$ for all voters with approval set $\{c_v,c_w,g_1\}$ where both $c_v\in S$ and $c_w\in S$.
	
	Now, it holds that
    \begin{align*}
        &\phantom{=:} 2m\delta(1) \\
        &= \alpha n^4\delta(1) - t\alpha n^3(\delta(1)-\delta(2)) + \frac 12\alpha n^3(\delta(1)-\delta(2))\\
        &\ge \alpha n^4\delta(1) - t\alpha n^3(\delta(1)-\delta(2)) + \frac 12 n^3\delta(1)\\
        & > \alpha n^4\delta(1) - t\alpha n^3(\delta(1)-\delta(2)) + 3\delta(1)\text.
    \end{align*}

    Once again, we use that $\alpha(\delta(1)-\delta(2))\ge \delta(1)$ and $n\ge 2$.
    Therefore, the next candidate to be selected is $g_1$.
	Then, the score of $g_4$ is largest, and $g_4$ is elected.
	
	We arrive at a point where $\sco{g_2}{S\cup \{b,g_1,g_4\}}$ and $\sco{g_3}{S\cup \{b,g_1,g_4\}}$ are bounded by 
 
    \begin{align*}
        &\phantom{=:}2 m \delta(2) + (|E| +2)\delta(1)\\
        &< 2 m \delta(2) + n^2\delta(1)\\
        & = \alpha n^4\delta(2) -\left(t-\frac 12\right) n^3\frac{\delta(1)-\delta(2)}{\delta(1)}\delta(2) + n^2\delta(1)\\
        &< \alpha n^4\delta(2) + n^2\delta(1)\\
        &= \alpha n^3(n-1)\delta(2) + \alpha n^3 \delta(1) \\\phantom{=}&+ \alpha n^3(\delta(2) - \delta(1)) + n^2\delta(1)\\
        &\le \alpha n^3(n-1)\delta(2) + \alpha n^3 \delta(1) - n^3\delta(1) + n^2\delta(1)\\
        & < \alpha n^3(n-1)\delta(2) + \alpha n^3 \delta(1)\text.
    \end{align*}

	However, while not all candidates of type $c_v$ are elected, the score of such a candidate is at least $\alpha n^3(n-1)\delta(2) + \alpha n^3\delta(1)$, where we just count the contribution of voters with approval sets $\{c_v,c_w\}$ and $\{c_v\}$.
	Hence, next, we have to select all candidates of type $c_v$.
	
	It remains to select the final candidate.
	Clearly, $\sco{g_3}{C\setminus \{g_2,g_3\}} = 2m\delta(2) < \sco{g_2}{C\setminus \{g_2,g_3\}}$, and $g_2$ is selected.
	
	Hence, the choice set contains exactly the committee $\{b\}\cup C_V\cup\{g_1,g_2,g_4\}$.

	Now, consider the situation where some voter with approval set $\{g_1,g_3\}$ abstains from the election.
	
	Until $\{b\}$ and a set $S\subseteq C_V$ of $t$ candidates are selected, the Thiele rule represented by $s$ proceeds identically.
	We arrive at a point where the scores of $g_2$ and $g_4$ are the same as without abstention, whereas 
    \begin{itemize}
        \item $\sco{g_3}{\{b\}\cup S} = (2m-1)\delta(1)$,
        \item $\sco{g_1}{\{b\}\cup S} = (2m + |E|-3t)\delta(1) + 3t\delta(2)$ if $S$ is an independent set, and
        \item $\sco{g_1}{\{b\}\cup S} = (2m + |E|-3t)\delta(1) + 3t\delta(2) + \left[(\delta(1) - \delta(2)) - (\delta(2)-\delta(3))\right]$ if $S$ is not an independent set.
    \end{itemize}
  
	Hence, if $S$ is not an independent set, then $g_1$ is still selected next, and the procedure continues as without abstention, yielding the committee $S\cup C_V\cup\{g_1,g_2,g_4\}$.
	
	However, if $S$ is an independent set, then $g_2$ has an equal score as $g_1$ and may be selected first.
	Similar arguments as before show that the subsequent selection order is $g_3$ next, then the remaining candidates from $C_V$, and finally $g_1$.
	
	Consequently, if the source instance is a No-instance, then the choice set is identical after abstention, and there is no incentive to abstain.
	Otherwise, the choice set additionally contains $\{b\}\cup C_V\cup\{g_1,g_2,g_3\}$ and is preferred by a voter with approval set $\{g_1,g_3\}$ (according to Kelly's extension).
\end{proof}

Next, we prove our theorem concerning {\mes} and {\sphrag}.

\hardMES*

We carry out the key technical challenge of the proof for {\mes} in a technical lemma.
The lemma also gives more insight about how the election with {\mes} actually selects candidates:
We provide a unified approach from which we obtain hardness of abstention for both Phase~1 and after the complete execution of {\mes} with a non-trivial Phase~1.
The assertion for {\mes} is a direct consequence.
In the sequel, we extract a proof for {\sphrag} from the same reduction.

\begin{restatable}{lemma}{MESreduction}\label{lem:MESred}
    Consider voting by {\mes} with completion by {\sphrag}. Then, there exists a polynomial-time reduction from an \NP-complete problem such that there exists a voter in the reduced instance with the following properties.
\begin{enumerate}
    \item If the source instance is a No-instance, then the outcome of the election is identical with and without this voter's abstention both after Phase~1 and after Phase~2.
    \item If the source instance is a Yes-instance, then the outcome of the election under {\mes} is preferred by this voter under abstention both after Phase~1 and after Phase~2.
\end{enumerate}
\end{restatable}

\begin{proof}
    We reduce from the \NP-complete problem \textsc{Restricted Exact Cover by 3-Sets (RX3C)} \citep{Gonz85a}. 
    An instance $(U, \mathcal S)$ consists of a finite universe $U=\{x_1,\dots, x_{3t}\}$ and a family $\mathcal{S}=\{S_1,\dots , S_{3t}\}$ of $3$-subsets of $U$ where each element from $U$ appears in exactly three sets from~$\mathcal{S}$. 
	The question is whether there is a family $\mathcal{S}'\subseteq \mathcal{S}$ which is an exact cover of $U$.
    For each $i\in [3t]$, let $\mathcal S_i = \{S\in \mathcal S\colon x_i\in S\}$.
    Moreover, we say that $\mathcal S'\subseteq \mathcal S$ is a \emph{maximal cover} if $(i)$ for all $S, T\in \mathcal S'$, $S\cap T = \emptyset$ and $(ii)$ for all $S\in \mathcal S\setminus \mathcal S'$, there exists $T\in \mathcal S'$ with $S\cap T\neq \emptyset$.
    In other words, $(U,\mathcal S)$ is a Yes-instance if and only if there exists a maximal cover of size $t$.
    Without loss of generality, we assume that 
    \begin{equation}\label{eq:t-bound}
        t^3 \ge 90 t^2 + 120t + 60\text.
    \end{equation}

    Assume that we are given an instance $(U, \mathcal S)$.
    We construct a reduced instance.

    First, the set of candidates is $C = \{g_i\colon i\in [6]\}\cup \{c_S\colon S\in \mathcal S\}\cup \{a_i\colon i\in [t]\}\cup \{b_1,b_2\}$.

    The role of the candidates is as follows:
    \begin{itemize}
        \item Candidates $g_i$ form a gadget in which abstention might be possible.
        \item Candidates $c_S$ represent the sets in the source instance.
        \item Candidates $a_i$ are auxiliary candidates who can help to create a certain symmetry between the gadget candidates $g_1$ and $g_4$.
        \item Candidates $b_1$ and $b_2$ are used to fill the committee in Phase~2.
    \end{itemize}

    Moreover, we have voters with the following approval sets:
    \begin{itemize}
        \item For each $i\in [t]$, there are $10t^3 + 45t$ voters with approval set $\{a_i\}$ and $15$ voters with approval set $\{g_4,a_i\}$.
        \item There are $7t^3$ voters approving $\{b_1\}$ and $7t^3 - (90 t^2 + 120t + 60)$ voters approving $\{b_2\}$. This is well defined because of \Cref{eq:t-bound}.
        \item For each $s\in S$, there are $10t^3$ voters with approval set $\{c_S\}$.
        \item For each $i\in [3t]$, there are $5$ voters with approval set $\{c_S\colon S\in \mathcal S_i\}\cup \{g_1\}$ and $15t$ voters with approval set $\{c_S\colon S\in \mathcal S_i\}$.
        \item One voter approving $\{g_1,g_2,g_3\}$. This voter will be used for the potential abstention.
        \item The following voters that approve only gadget candidates:
        \begin{itemize}
            \item $4t^3 + 30t + 13$ voters approving $\{g_1,g_4\}$,
            \item $6t^3$ voters approving $\{g_1\}$,
            \item $3t^3 + 11$ voters for each of the approval sets $\{g_2\}$ and $\{g_3\}$,
            \item $12$ voters for each of the approval sets $\{g_5\}$ and $\{g_6\}$,
            \item $7t^3 + 30t$ voters for each of the approval sets $\{g_2,g_5\}$ and $\{g_3,g_6\}$, and
            \item $3t^3$ voters for each of the approval sets $\{g_4,g_5\}$ and $\{g_4,g_6\}$.
        \end{itemize}
    \end{itemize}
    
    \begin{figure}
        \centering
        \resizebox{1\columnwidth}{!}{
        \begin{tikzpicture}
            \node (g1) at (0,6)  {};
            \node (g2) at (-3,5) {};
            \node (g3) at (3,5)  {};
            \node (g4) at (0,0)  {};
            \node (g5) at (-3,2) {};
            \node (g6) at (3,2)  {};

            \draw[TUMblue, fill = TUMblue, fill opacity = .2] (g1) ellipse (1.75cm and .8cm);
            \draw[TUMblue, fill = TUMblue, fill opacity = .2] ($(g2)+(-.45,0)$) ellipse (1.75cm and .8cm);
            \draw[TUMblue, fill = TUMblue, fill opacity = .2] ($(g3)+(.45,0)$) ellipse (1.75cm and .8cm);
            \draw[TUMblue, fill = TUMblue, fill opacity = .2] ($(g5)+(-.45,0)$) ellipse (1.75cm and .8cm);
            \draw[TUMblue, fill = TUMblue, fill opacity = .2] ($(g6)+(.45,0)$) ellipse (1.75cm and .8cm);
            \draw[TUMblue, fill = TUMblue, fill opacity = .2] ($(g1)!.5!(g4)$) ellipse (.8cm and 4cm);
            \draw[TUMblue, fill = TUMblue, fill opacity = .2] ($(g2)!.5!(g5)$) ellipse (.8cm and 2.5cm);
            \draw[TUMblue, fill = TUMblue, fill opacity = .2] ($(g3)!.5!(g6)$) ellipse (.8cm and 2.5cm);
            \draw[rotate = -56, TUMblue, fill = TUMblue, fill opacity = .2] ($(g4)!.5!(g6)$) ellipse (.8cm and 2.5cm);
            \draw[rotate = 56, TUMblue, fill = TUMblue, fill opacity = .2] ($(g4)!.5!(g5)$) ellipse (.8cm and 2.5cm);
            \draw[ultra thick,TUMgreen] \convexpath{g1, g3, g2}{0.7cm};

            \node[align = center] at ($(g1)!.5!(g4)$) {$4t^3 + 30t$\\$+13$};
            \node at ($(g1)+(1.15,0)$)   {$6t^3$};
            \node at ($(g2)+(-1.4,0)$)   {$3t^3+11$};
            \node at ($(g3)+(1.4,0)$)   {$3t^3+11$};
            \node at ($(g5)+(-1.2,0)$)   {$12$};
            \node at ($(g6)+(1.2,0)$)   {$12$};
            \node at ($(g2)!.5!(g5)$) {$7t^3 + 30t$};
            \node at ($(g3)!.5!(g6)$) {$7t^3 + 30t$};
            \node at ($(g5)!.5!(g4)$) {$3t^3$};
            \node at ($(g6)!.5!(g4)$) {$3t^3$};
            \foreach \i in {1,2,3,4,5,6}
            \node[circle, draw,fill = white] at (g\i) {$g_\i$};
        \end{tikzpicture}
        } 
        \caption{Approval ballots of the gadget candidates. The candidates $g_i$ are the white circles. Each ellipse displays a set of voters with the indicated multiplicity. The green shape indicates the potential abstaining agent with approval set $\{g_1,g_2,g_3\}$.}
        \label{fig:MESgadget}
    \end{figure}
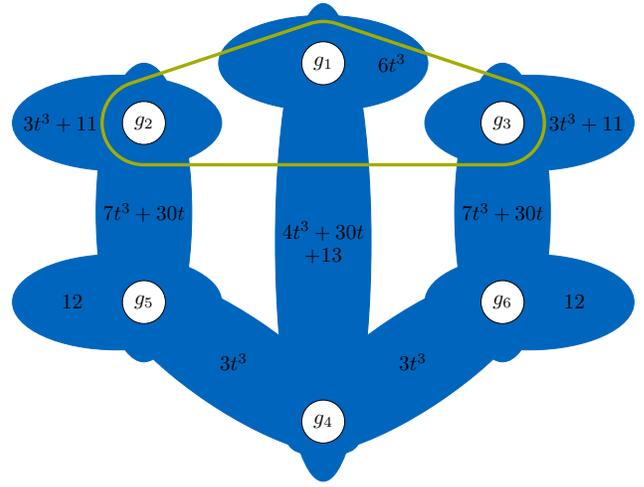
    
    The target committee size is $4t + 5$. 
    
    Note that there is a total of $10t^3(4t+5)$ voters.
    Hence, every voter has an initial budget of $1/(10t^3)$, and it is essentially sufficient for a candidate to be added to the committee that $10t^3$ voters which have not spent budget thus far approve her.
    The approval scores of the candidates are:
    \begin{itemize}
        \item $10t^3 + 45t + 15$ for candidates $a_i$ and $c_S$
        \item $10t^3 + 45t + 14$ for $g_1$,
        \item $10t^3 + 45t + 13$ for $g_4$,
        \item $10t^3 + 30t + 12$ for $g_2$, $g_3$, $g_5$ and $g_6$,
        \item $7t^3$ for $b_1$, and 
        \item $7t^3 - (90 t^2 + 120t + 60)$ for $b_2$.
    \end{itemize}

    The instance is designed in such a way that the execution of {\mes} is identical to an execution of {\sccav}: Whenever a candidate is bought, then the voters approving this candidate have to spend almost their entire budget, and hence the next candidate is a candidate which maximizes the approval score of the voters without an approved candidate in the committee.

    The most relevant overlapping of approval sets concerns gadget agents. The approval scores of the gadget candidates are designed in a way such that either the agents $\{g_1,g_5,g_6\}$ or (and this will only happy under abstention if the source instance is a Yes-instance) the agents $\{g_4,g_2,g_3\}$ are selected.
    The approval sets for the voters that only approve gadget agents are displayed in \Cref{fig:MESgadget} and illustrate this interplay. For instance, if $g_1$ is selected first, then there is not enough budget left to buy $g_4$ in Phase~1.

    Let us consider the exectuation of {\mes} for the reduced instance.
    The first candidates to be selected are of type $a_i$ or $c_S$. 
    If a candidate $a_i$ is selected, then the budget of voters approving candidates $a_i$ or $c_S$ are not affected. 
    If $c_S$ is selected, then the budget of a voter who approves $c_S$ is reduced to $\frac 1{10t^3} - \frac 1 {10t^3 + 45t + 15} = \frac {45t + 15}{10t^3(10t^3+45t + 15)}$, which is negligible for buying further candidates in Phase~1.
    In particular, for $T\in \mathcal S$ with $T\cap S\neq \emptyset$, there are at most $10t^3+30t+10$ supporters of $c_T$ with full budget left. 
    Hence, {\mes} has to select $g_1$ or $g_4$ before candidates $c_T$.

    However, as long as a set $S\in \mathcal S$ has an empty intersection with all sets $T$ of candidates $c_T$ selected thus far, $S$ still is approved by $10t^3 + 45t + 15$ voters with full budget.

    As a consequence, we select first all candidates of type $a_i$ as well as candidates $\{c_S\colon S\in \mathcal S'\}$ for some maximal cover $\mathcal S'\subseteq \mathcal S$. 

    We reach a point, where each candidate $c_S$ for $S\in \mathcal S\setminus \mathcal S'$ is only approved by at most $10t^3 + 30t + 10$ voters with full budget, whereas all other voters approving $c_S$ have negligible budget.

    At this point, the only voters approving a gadget candidate that already have spent budget approve $g_1$ or $g_4$.
    For $g_4$, there are still $10t^3+30t+13$ voters with full budget, but the voters with approval sets $\{a_i, g_4\}$ have spent most of their budget. Hence, there are $15t$ voters that have only left a budget of $\frac {45t + 15}{10t^3(10t^3+45t + 15)}$.
    Next, we consider $g_1$. Therefore, let $\ell := |\mathcal S'|$. 
    Note that $\ell \le t$. Then, there are $10t^3 + 30t + 14 + 15(t-\ell)$ voters with full budget and $15\ell$ voters with a remaining budget of $\frac {45t + 15}{10t^3(10t^3+45t + 15)}$.

    Hence, $g_1$ will be bought next. This will in particular cause that $g_4$ cannot be bought in Phase~1 anymore. Now, the one voter with approval set $\{g_1,g_2,g_3\}$ lowers the number of full budget voters of $g_2$ and $g_3$, and we have to buy $g_5$ and $g_6$ next.

    Now, all the voters supporting gadget candidates have negligible budget. Hence, we complete Phase~1 of {\mes} by buying all remaining candidates of type $c_S$. This is possible because each of them has $10t^3$ supporters only approving them.
    At this point, the remaining total budget for candidates of type $a_i$, $c_S$, or $g_i$ is at most $(10t^3 + 45t + 15)\frac {45t + 15}{10t^3(10t^3+45t + 15)} = \frac {45t + 15}{10t^3}$.
    However, $b_1$ and $b_2$ still accumulate a total budget of $7t^3$ and $7t^3  - (90 t^2 + 120t + 60)\ge 6t^3$ (see \Cref{eq:t-bound}), respectively.
    Hence, in the completion phase starting with these budgets, {\sphrag} selects $b_1$ and $b_2$.

    Together, the unique selected committee is 
    \begin{align*}
        &\{c_S\colon S\in \mathcal S\}\cup \{a_i\colon i\in [t]\}\cup\{b_1,b_2,g_1,g_5,g_6\}\text,
    \end{align*}
    where $b_1$ and $b_2$ are not present, yet, at the end of Phase~1.

    Now, consider the situation where the voter with approval set $\{g_1,g_2,g_3\}$ abstains from the election.
    Now, voters only have a slightly larger budget of $\frac{4t+5}{10t^3(4t-5) -1}$, but the influence of this slight increase is negligible for the outcome of the election.
    We still select first candidates $a_i$ and $\{c_S\colon S\in \mathcal S'\}$ for some maximal cover $\mathcal S'\subseteq \mathcal S$.
    However, we now reach a stage where it may be possible to select $g_4$ instead.

    Let again $\ell := |\mathcal S'|$. At this point, for $g_4$, there are still $10t^3+30t+13$ voters with full budget and $15t$ voters with negligible budget, whereas $g_1$ is approved by $10t^3 + 30t + 13 + 6(t-\ell)$ voters with full budget and $15\ell$ voters with a negligible budget. Note that the negligible budget is of exactly the same size for voters approving $g_1$ and $g_4$.
     Hence, it is possible to select $g_4$ next if and only if $\mathcal S'$ is a maximal cover of size $t$ which happens if and only if the source instance is a Yes-instance.
     If $g_4$ is selected, then $g_2$ and $g_3$ are selected next. Then, Phase~1 is completed by selecting the remaining candidates of type $c_S$, and in Phase~2, $b_1$ and $b_2$ are selected.
     If $g_1$ is selected instead of $g_4$, then the election proceeds as without abstention.

    Hence, if the source instance is a Yes-instance, then the committee
        \begin{align*}
        &\{c_S\colon S\in \mathcal S\}\cup \{a_i\colon i\in [t]\}\cup\{b_1,b_2,g_4,g_2,g_3\}
    \end{align*}

    is also selected under abstention.

     Hence, the abstaining agent benefits from abstention (according to Kelly's extension) if and only if the source instance is a Yes-instance.
\end{proof}

We complete the proof of \Cref{thm:hardMES} by showing the hardness result for {\sphrag} in the next theorem.
Interestingly, the same reduction as in the previous proof also works for {\sphrag}.

\begin{restatable}{theorem}{PhragmenNP}
    For {\sphrag}, it is \NP-hard to decide whether a voter can benefit from abstention.
\end{restatable}

\begin{proof}
    We can use exactly the same reduction as in the proof of \Cref{lem:MESred}.
    We briefly outline, how the election with {\sphrag} proceeds on the reduced instance.

    The first time where voters have accumulated enough money to buy an alternative is at time $\frac 1{10t^3 + 45t + 15}$.
    At this time, {\sphrag} simultaneously buys all candidates of type $a_i$ as well as candidates $\{c_S\colon S\in \mathcal S'\}$ for some maximal cover $\mathcal S'\subseteq \mathcal S$.

    After these candidates are bought, at most $10t^3 + 30t + 10$ voters approving $c_T$ for $T\notin \mathcal S'$ have their money left, but there are still at least $10t^3 + 30t + 11$ voters approving each of the gadget candidates.
    Hence, at least at time $\frac 1{10t^3 + 30t + 11}$, we have enough budget to buy gadget candidates.
    Note that at this point, assuming that $t$ is large enough, the slightly more voters approving candidates of type $c_T$ have not accumulated enough money to buy further such candidates.

    Hence, until time $\frac 1{10t^3 + 30t + 11}$, we buy $\{g_1,g_5,g_6\}$ if no voter abstains or $\mathcal S'$ is of size less than $t$.
    If the voter approving $\{g_1,g_2,g_3\}$ abstains, then any of $g_1$ and $g_4$ may be bought at the same time, namely at time $\frac 1{10t^3 + 30t + 13}$, but only if $\mathcal S'$ is of size exactly $t$.
    Also, the accumulated budget only suffices to buy one of them.
    This offers the possibility to buy $g_4$ first and then $g_2$ and $g_3$ afterwards if the source instance is a Yes-instance.

    After the gadget candidates have been bought, we buy the remaining candidates of type $c_S$ latest at time  $\frac 1{10t^3 + 15t}$.
    In principal, we could omit the auxiliary candidates $b_1$ and $b_2$ from the reduction and have a committee of size $2$ less compared to the election with {\mes}. If we want to maintain them in the reduction, then they are the last candidates to be bought latest at time  $\frac 1{7t^3 - (90t^2 + 120t + 60)}$.

    Hence, the election with {\sphrag} leads to the exact same outcome as with {\mes} and we conclude that the voter with approval set $\{g_1,g_2,g_3\}$ can improve her outcome by abstention (according to Kelly's extension) if and only if the source instance was a Yes-instance.
\end{proof}

Finally, we consider the corollary for the verification and approval guarantee on selected committees.

\Verification*
\begin{proof}
    The result directly follows from the reductions in the proofs of \Cref{thm:SeqThieleNP,thm:hardMES}.
    For the first result, we omit the abstaining agent and question whether the committee that is always selected is the unique winner. For the second result, we want the abstaining agent to be present, so we omit the critical gadget alternative from some other voter approving gadget candidates.
\end{proof}


\begin{thebibliography}{38}
\providecommand{\natexlab}[1]{#1}

\bibitem[{Aziz et~al.(2015)Aziz, Gaspers, Gudmundsson, Mackenzie, Mattei, and
  Walsh}]{AGG+15a}
Aziz, H.; Gaspers, S.; Gudmundsson, J.; Mackenzie, S.; Mattei, N.; and Walsh,
  T. 2015.
\newblock Computational Aspects of Multi-Winner Approval Voting.
\newblock In \emph{Proceedings of the 14th International Conference on
  Autonomous Agents and Multiagent Systems (AAMAS)}, 107--115.

\bibitem[{Botan(2021)}]{Bota21a}
Botan, S. 2021.
\newblock Manipulability of Thiele Methods on Party-List Profiles.
\newblock In \emph{Proceedings of the 20th International Conference on
  Autonomous Agents and Multiagent Systems (AAMAS)}, 223--231.

\bibitem[{Brandl et~al.(2019)Brandl, Brandt, Geist, and Hofbauer}]{BBGH18a}
Brandl, F.; Brandt, F.; Geist, C.; and Hofbauer, J. 2019.
\newblock Strategic Abstention based on Preference Extensions: {P}ositive
  Results and Computer-Generated Impossibilities.
\newblock \emph{Journal of Artificial Intelligence Research}, 66: 1031--1056.

\bibitem[{Brandl, Brandt, and Hofbauer(2019)}]{BBH15c}
Brandl, F.; Brandt, F.; and Hofbauer, J. 2019.
\newblock Welfare Maximization Entices Participation.
\newblock \emph{Games and Economic Behavior}, 14: 308--314.

\bibitem[{Brandt, Geist, and Peters(2017)}]{BGP16c}
Brandt, F.; Geist, C.; and Peters, D. 2017.
\newblock Optimal Bounds for the No-Show Paradox via {SAT} Solving.
\newblock \emph{Mathematical Social Sciences}, 90: 18--27.
\newblock Special Issue in Honor of Herv\'e Moulin.

\bibitem[{Bredereck et~al.(2021)Bredereck, Faliszewski, Kaczmarczyk,
  Niedermeier, Skowron, and Talmon}]{BFK+21a}
Bredereck, R.; Faliszewski, P.; Kaczmarczyk, A.; Niedermeier, R.; Skowron, P.;
  and Talmon, N. 2021.
\newblock Roubstness among multiwinner voting rules.
\newblock \emph{Artificial Intelligence}, 290: 103403.

\bibitem[{Brill et~al.(2023)Brill, Dindar, Israel, Lang, Peters, and
  {Schmidt-Kraepelin}}]{BDI+23a}
Brill, M.; Dindar, H.; Israel, J.; Lang, J.; Peters, J.; and
  {Schmidt-Kraepelin}, U. 2023.
\newblock Multiwinner Voting with Possibly Unavailable Candidates.
\newblock In \emph{Proceedings of the 37th AAAI Conference on Artificial
  Intelligence (AAAI)}.
\newblock Forthcoming.

\bibitem[{Brill et~al.(2017)Brill, Freeman, Janson, and Lackner}]{BFJL16a}
Brill, M.; Freeman, R.; Janson, S.; and Lackner, M. 2017.
\newblock Phragm\'{e}n's Voting Methods and Justified Representation.
\newblock In \emph{Proceedings of the 31st AAAI Conference on Artificial
  Intelligence (AAAI)}, 406--413.

\bibitem[{Brill, Laslier, and Skowron(2018)}]{BLS18a}
Brill, M.; Laslier, J.-F.; and Skowron, P. 2018.
\newblock Multiwinner Approval Rules as Apportionment Methods.
\newblock \emph{Journal of Theoretical Politics}, 30(3): 358--382.

\bibitem[{Cevallos and Stewart(2021)}]{CeSt21a}
Cevallos, A.; and Stewart, A. 2021.
\newblock A verifiably secure and proportional committee election rule.
\newblock In \emph{Proceedings of the 3rd ACM Conference on Advances in
  Financial Technologies}, 29--42.

\bibitem[{Delemazure et~al.(2023)Delemazure, Demeulemeester, Eberl, Israel, and
  Lederer}]{DDE+22a}
Delemazure, T.; Demeulemeester, T.; Eberl, M.; Israel, J.; and Lederer, P.
  2023.
\newblock Strategyproofness and Proportionality in Party-approval Multiwinner
  Elections.
\newblock In \emph{Proceedings of the 37th AAAI Conference on Artificial
  Intelligence (AAAI)}, 5591--5599.

\bibitem[{Dong and Lederer(2023)}]{DoLe22a}
Dong, C.; and Lederer, P. 2023.
\newblock Characterizations of Sequential Valuation Rules.
\newblock In \emph{Proceedings of the 22nd International Conference on
  Autonomous Agents and Multiagent Systems (AAMAS)}, 1697--1705.

\bibitem[{Dong and Lederer(2024)}]{DoLe23a}
Dong, C.; and Lederer, P. 2024.
\newblock Refined Characterizations of Approval-Based Committee Scoring Rules.
\newblock In \emph{Proceedings of the 38th AAAI Conference on Artificial
  Intelligence (AAAI)}.
\newblock Forthcoming.

\bibitem[{Duddy(2014)}]{Dudd14b}
Duddy, C. 2014.
\newblock Condorcet's principle and the strong no-show paradoxes.
\newblock \emph{Theory and Decision}, 77(2): 275--285.

\bibitem[{Faliszewski, Gawron, and Kusek(2022)}]{FGK22b}
Faliszewski, P.; Gawron, G.; and Kusek, B. 2022.
\newblock Using Multiwinner Voting to Search for Movies.
\newblock In \emph{Proceedings of the 19th European Conference on Multi-Agent
  Systems (EUMAS)}, Lecture Notes in Computer Science (LNCS), 116--133.
  Springer-Verlag.

\bibitem[{Faliszewski et~al.(2017)Faliszewski, Skowron, Slinko, and
  Talmon}]{FSST17a}
Faliszewski, P.; Skowron, P.; Slinko, A.; and Talmon, N. 2017.
\newblock Multiwinner Voting: A New Challenge for Social Choice Theory.
\newblock In Endriss, U., ed., \emph{Trends in Computational Social Choice},
  chapter~2.

\bibitem[{Fishburn(1972)}]{Fish72a}
Fishburn, P.~C. 1972.
\newblock Even-chance lotteries in social choice theory.
\newblock \emph{Theory and Decision}, 3(1): 18--40.

\bibitem[{G{\"a}rdenfors(1976)}]{Gard76a}
G{\"a}rdenfors, P. 1976.
\newblock Manipulation of Social Choice Functions.
\newblock \emph{Journal of Economic Theory}, 13(2): 217--228.

\bibitem[{Garey and Johnson(1979)}]{GaJo79a}
Garey, M.~R.; and Johnson, D.~S. 1979.
\newblock \emph{Computers and Intractability: A Guide to the Theory of
  NP-Completeness}.
\newblock W. H. Freeman.

\bibitem[{Gawron and Faliszewski(2022)}]{GaFa22a}
Gawron, G.; and Faliszewski, P. 2022.
\newblock Using Multiwinner Voting to Search for Movies.
\newblock In \emph{Proceedings of the 19th European Conference on Multi-Agent
  Systems (EUMAS)}, Lecture Notes in Computer Science (LNCS), 134--151.
  Springer-Verlag.

\bibitem[{Hofbauer(2019)}]{Hofb19a}
Hofbauer, J. 2019.
\newblock \emph{Should {I} Stay or Should {I} Go? {T}he No-Show Paradox in
  Voting and Assignment}.
\newblock Ph.D. thesis, Technische Universit{\"a}t M{\"u}nchen.

\bibitem[{Janeczko and Faliszewski(2023)}]{JaFa23b}
Janeczko, L.; and Faliszewski, P. 2023.
\newblock Ties in Multiwinner Approval Voting.
\newblock In \emph{Proceedings of the 32nd International Joint Conference on
  Artificial Intelligence (IJCAI)}, 2765--2773.

\bibitem[{Janson(2016)}]{Jans16a}
Janson, S. 2016.
\newblock Phragm{\'e}n's and {T}hiele's election methods.
\newblock Technical Report arXiv:1611.08826 [math.HO], arXiv.org.

\bibitem[{Jimeno, P\'erez, and Garc\'ia(2009)}]{JPG09a}
Jimeno, J.~L.; P\'erez, J.; and Garc\'ia, E. 2009.
\newblock An extension of the {M}oulin {N}o {S}how {P}aradox for voting
  correspondences.
\newblock \emph{Social Choice and Welfare}, 33(3): 343--459.

\bibitem[{Kelly(1977)}]{Kell77a}
Kelly, J.~S. 1977.
\newblock Strategy-Proofness and Social Choice Functions Without
  Single-Valuedness.
\newblock \emph{Econometrica}, 45(2): 439--446.

\bibitem[{Kluiving et~al.(2020)Kluiving, {de Vries}, Vrijbergen, Boixel, and
  Endriss}]{KdV+20a}
Kluiving, B.; {de Vries}, A.; Vrijbergen, P.; Boixel, A.; and Endriss, U. 2020.
\newblock Analysing Irresolute Multiwinner Voting Rules with Approval Ballots
  via {SAT} Solving.
\newblock In \emph{Proceedings of the 24th European Conference on Artificial
  Intelligence (ECAI)}.

\bibitem[{Lackner and Skowron(2021)}]{LaSk21a}
Lackner, M.; and Skowron, P. 2021.
\newblock Consistent Approval-Based Multi-Winner Rules.
\newblock \emph{Journal of Economic Theory}, 192: 105173.

\bibitem[{Lackner and Skowron(2023)}]{LaSk22b}
Lackner, M.; and Skowron, P. 2023.
\newblock \emph{Multi-Winner Voting with Approval Preferences}.
\newblock Springer-Verlag.

\bibitem[{Mora and Oliver(2015)}]{MoOl15a}
Mora, X.; and Oliver, M. 2015.
\newblock Eleccions mitjan{\c c}ant el vot d'aprovaci{\'o}. {E}l m{\`e}tode de
  {P}hragm{\'e}n i algunes variants.
\newblock \emph{Butllet{\'\i} de la Societat Catalana de Matem{\`a}tiques},
  30(1): 57--101.

\bibitem[{Moulin(1988)}]{Moul88b}
Moulin, H. 1988.
\newblock Condorcet's Principle implies the No Show Paradox.
\newblock \emph{Journal of Economic Theory}, 45(1): 53--64.

\bibitem[{P\'erez(2001)}]{Pere01a}
P\'erez, J. 2001.
\newblock The {S}trong {N}o {S}how {P}aradoxes are a common flaw in {C}ondorcet
  voting correspondences.
\newblock \emph{Social Choice and Welfare}, 18(3): 601--616.

\bibitem[{P{\'e}rez, Jimeno, and Garc{\'\i}a(2010)}]{PJG10a}
P{\'e}rez, J.; Jimeno, J.~L.; and Garc{\'\i}a, E. 2010.
\newblock No Show Paradox in Condorcet k-voting Procedure.
\newblock \emph{Group Decision and Negotiation}, 21(3): 291--303.

\bibitem[{Peters(2018)}]{Pete18a}
Peters, D. 2018.
\newblock Proportionality and Strategyproofness in Multiwinner Elections.
\newblock In \emph{Proceedings of the 17th International Conference on
  Autonomous Agents and Multiagent Systems (AAMAS)}, volume 1549--1557.

\bibitem[{Peters and Skowron(2020)}]{PeSk20a}
Peters, D.; and Skowron, P. 2020.
\newblock Proportionality and the Limits of Welfarism.
\newblock In \emph{Proceedings of the 21nd ACM Conference on Economics and
  Computation (ACM-EC)}, 793--794.

\bibitem[{Phragm{\'e}n(1895)}]{Phra95a}
Phragm{\'e}n, E. 1895.
\newblock \emph{Proportionella val. En valteknisk studie}.
\newblock Svenska sp{\"o}rsm{\aa}l 25. Lars H{\"o}kersbergs f{\"o}rlag,
  Stockholm.

\bibitem[{S\'{a}nchez-Fern\'andez and Fisteus(2019)}]{SaFi19a}
S\'{a}nchez-Fern\'andez, L.; and Fisteus, J.~A. 2019.
\newblock Monotonicity axioms in approval-based multi-winner voting rules.
\newblock In \emph{Proceedings of the 18th International Conference on
  Autonomous Agents and Multiagent Systems (AAMAS)}, 485--493.

\bibitem[{Thiele(1895)}]{Thie95a}
Thiele, T.~N. 1895.
\newblock Om Flerfoldsvalg.
\newblock \emph{Oversigt over det Kongelige Danske Videnskabernes Selskabs
  Forhandlinger}, 415--441.

\bibitem[{Young(1975)}]{Youn75a}
Young, H.~P. 1975.
\newblock Social Choice Scoring Functions.
\newblock \emph{SIAM Journal on Applied Mathematics}, 28(4): 824--838.

\end{thebibliography}
\end{document}